\documentclass[twocolumn,10pt]{IEEEtran}
\usepackage{amsmath,amsthm}
\usepackage{mathtools}
\mathtoolsset{mathic}
\usepackage{microtype}
\usepackage{cite}

\usepackage{ifxetex,ifluatex}
\ifxetex
\usepackage{xltxtra} 
\else
\ifluatex
\usepackage{luatextra}
\else
\usepackage{amssymb}
\usepackage[utf8]{inputenc}
\csname else\expandafter\endcsname
\romannumeral -`0
\fi
\fi
\iftrue\relax%
\defaultfontfeatures{Ligatures=TeX}
\usepackage[math-style=TeX, vargreek-shape=TeX]{unicode-math}
\fi
\usepackage{algorithmic, algorithm}

\usepackage{tikz}
\usetikzlibrary{shapes.geometric, positioning, shadows}
\usepackage{graphicx,color}

\usepackage{xparse}
\usepackage[bookmarksnumbered,bookmarksopen,
unicode]{hyperref}
\hypersetup{
  pdftitle={Information Theoretic Sequential Adaptive Compressed Sensing},
  pdfauthor={Gábor Braun, Sebastian Pokutta, Yao C.~Xie},
  pdfkeywords={},
  pdfsubject={MSC2000 Primary ;
    Secondary}} 

\newcommand{\R}{\mathbb{R}}
\newcommand{\N}{\mathbb{N}}

\DeclareMathOperator*{\argmax}{arg max}

\newcommand{\face}[1]{\left\{#1\right\}}
\newcommand{\card}[1]{\left|#1\right|}

\newcommand{\tnorm}[1]{\left\|#1\right\|_2}

\newcommand{\binSet}{\{0,1\}}

\newcommand{\yao}[1]{{\color{black}{#1}}}

\DeclarePairedDelimiter{\abs}{|}{|}
\DeclarePairedDelimiter{\norm}{\|}{\|}

\DeclarePairedDelimiter{\ceil}{\lceil}{\rceil}

\newcommand*{\transpose}[1]{#1\sp{\intercal}}

\newcommand*{\Dnormal}[2]{\mathcal{N}\left(#1,#2\right)}

\NewDocumentCommand{\probability}{d()om}{%
  \operatorname{\mathbb{P}}%
  \IfValueT{#1}{\sb{#1}}%
  \left[#3%
    \IfValueT{#2}{\,\middle|\,#2}\right]}
\NewDocumentCommand{\expectation}{d()om}%
  {\operatorname{\mathbb{E}}%
      \IfValueT{#1}{\sb{#1}}%
      \left[#3%
    \IfValueT{#2}{\,\middle|\,#2}\right]}
\NewDocumentCommand{\entropy}{om}{\mathbb{H}\left[#2
    \IfValueT{#1}{\,\middle|\,#1}\right]}
\NewDocumentCommand{\bentropy}{lm}
  {\widetilde{\mathbb{H}}#1\left[#2\right]}

\NewDocumentCommand{\mutualInfo}{omm}{\mathbb{I}\left[#2;#3
    \IfValueT{#1}{\,\middle|\,#1}\right]}

\theoremstyle{plain}
\newtheorem{theorem}{Theorem}[section]
\newtheorem{thm}[theorem]{Theorem}

\newtheorem{lem}[theorem]{Lemma}
\newtheorem{cor}[theorem]{Corollary}

\theoremstyle{definition}

\newtheorem{defn}[theorem]{Definition}

\theoremstyle{remark}

\newtheorem{rem}[theorem]{Remark}

\newtheoremstyle{claim}{}{}{}{}{\itshape}{.}{.5em}{}
\theoremstyle{claim}

  {\end{proof}}

\title{Info-Greedy Sequential Adaptive \\Compressed Sensing}

\author{Gábor Braun,\thanks{Gábor Braun
    (Email: gabor.braun@isye.gatech.edu),
    Sebastian Pokutta
    (Email: sebastian.pokutta@isye.gatech.edu)
    and
    Yao Xie (Email: yao.xie@isye.gatech.edu)
    are with the H. Milton Stewart School of
    Industrial and Systems Engineering, Georgia Institute of
    Technology, Atlanta, GA.}\quad
  \and Sebastian Pokutta, \and \quad Yao Xie
  \thanks{This work is partially supported by NSF grant CMMI-1300144
and CCF-1442635. Authors contributed equally to the paper.}
} \date{\today}

\newcommand{\yx}[1]{{\color{black}{#1}}}

\begin{document}
\maketitle

\begin{abstract}
  We present an information-theoretic framework for sequential
  adaptive compressed sensing, \emph{Info-Greedy Sensing}, where
  measurements are chosen to maximize the extracted information
  conditioned on the previous measurements. We show that the widely used bisection approach is Info-Greedy for a family of $k$-sparse signals by connecting compressed sensing and blackbox complexity of sequential query algorithms, and present Info-Greedy algorithms for Gaussian and Gaussian Mixture Model (GMM) signals, as well as ways to design sparse Info-Greedy measurements. Numerical
  examples demonstrate the good performance of the proposed algorithms
  using simulated and real data: 
  \yx{Info-Greedy Sensing shows significant improvement over random
    projection for signals with sparse and low-rank covariance matrices, and adaptivity brings robustness when there is a mismatch between the assumed and the true distributions.}
\end{abstract}

\section{Introduction}
\label{sec:introduction}

\yx{ Nowadays ubiquitous big data applications (image processing
  \cite{Brady2009}, power network monitoring
  \cite{WirelessHouseElectricity2014}, and large scale sensor networks
  \cite{sparseSensorLocal2011}) call for more efficient information
  \emph{sensing} techniques. Often these techniques are
  \emph{sequential} in that the measurements are taken one after
  another.  Hence information gained in the past can be used to guide
  an \emph{adaptive} design of subsequent measurements, which
  naturally leads to the notion of sequential adaptive sensing.  At
  the same time, a path to efficient sensing of big data is
  \emph{compressive sensing}
  \cite{CandesTao2006,DonohoCS2006,EldarKutinoik2012}, which exploits
  low-dimensional structures to recover signals from a number of
  measurements much smaller than the ambient dimension of the
  signals.}

Early compressed sensing works mainly focus on non-adaptive and
one-shot measurement schemes.
\yx{Recently there has also been much interest in
  \emph{sequential adaptive compressed sensing}, which measures noisy
  linear combinations of the entries (this is different from the
  direct adaptive sensing, which measures signal entries directly
\cite{HauptDistill2011,WeiHero2013JSec,WeiHero2013,MalloyNowakSeqTest2012}).}
Although in the seminal work of \cite{CastroCandesDavenport2012}, it
was shown under fairly general assumptions that ``adaptivity does not
help much'', i.e., sequential adaptive compressed sensing does not
improve the \emph{order} of the min-max bounds obtained by algorithms,
these limitations are restricted to certain performance metrics. It
has also been recognized (see, e.g.,
\cite{IndykPrice2011, MalloyNowak2013, Info_theo_CS_bound2014}) that
adaptive compressed sensing offers several benefits with respect to
other performance metrics, such as the reduction in the
signal-to-noise ratio (SNR) to recover the signal. Moreover, larger
performance gain can be achieved by adaptive compressed sensing if we
aim at recovering a ``family'' of signals with known statistical
\emph{prior information} (incorporating statistical priors in
compressed sensing has been considered in \cite{CS_GMM_2011} for the
non-sequential
\yx{setting and in \cite{BayesianCS2008} for the sequential setting
  using Bayesian methods}).

\yx{To harvest the benefits of adaptive compressed sensing, } various
algorithms have been developed: compressive binary search
\cite{HauptAdaptiveCS2009, DavenportArias-Castro2012}, which considers
a problem of determining the location of a single non-zero entry; a
variant of the iterative bisection algorithm \cite{TajerPoor2012} to
adaptively identify the partial support of the signal;
random choice of compressed
sensing vectors \cite{MalioutovSanghaviWillsky2010}, and
a collection of independent structured random sensing matrices in each
measurement step  \cite{HauptSeqCS2012}
with some columns
``masked'' to zero;
\yx{an experimental design approach \cite{JainSoniHaupt2013} that designs measurements}
adaptive to the mean square error of the estimated signal;
 exploiting additional graphical
structure of the signal
\cite{KrishnamurthySingh13,ErvinCastro2013}; the CASS algorithm \cite{MalloyNowak2013},
which is based on bisection search to locate multiple non-zero
entries, and is
claimed to be near-optimal in the number of measurements needed
sequentially to achieve small recovery errors;
\yx{an
  adaptive sensing strategy specifically tailored to tree-sparse
  signals \cite{AkshayHaupt2014}
that significantly outperforms non-adaptive sensing
  strategies.
}
In optics literature, compressive imaging systems with  sequential measurement architectures have been developed \cite{NeifeldTaskSpecific2008, KeAshok2010, NeifeldCSImaging2011}, which may modify the measurement basis based on specific object information derived from the previous measurements and achieve better performance. \yao{In medical imaging literature, \cite{Opt_k_sparse} uses Bayesian experimental design to optimize $k$-space sampling for nonlinear sparse MRI reconstruction. }

The idea of using an information measure for sequential compressed
  sensing has been spelled out in various places for specific settings
  or signal models, for example, the seminal Bayesian compressive
  sensing work \cite{BayesianCS2008}, which designs a new projection
  that minimizes the differential entropy of the posterior estimate on
  a Gaussian signal; \cite[Chapter 6.2]{EldarKutinoik2012} and
  \cite{WaeberFrazier2013}, which introduces the so-called ``expected
  information'' and outlines a general strategy for sequential
  adaptive sensing; \cite{TaskDrivenDuarte2013}, which develops a
  two-step adaptive statistical compressed sensing scheme for Gaussian
  mixture model (GMM) signals based on maximizing an
  information-theoretic objective function;
  \cite{CarsonChenRodrigues2012}, which sequentially senses low-rank
  GMM signals based on a posterior distribution and provides an
  empirical performance analysis; \yao{\cite{ADesignPoisson2013} studies the design of linear projection measurements for a vector Poisson signal model; \cite{ICMLNonlinear2014} designs general nonlinear functions for mapping high-dimensional data into lower-dimensional space using mutual information as a metric.} A general belief, though, is that
  it is difficult to devise quantitative error bounds for such
  sequential information maximizing algorithms (see, e.g.,
  \cite[Section 6.2.3]{EldarKutinoik2012}).

\yx{In this work, we present a unified information theoretical
  framework for sequential adaptive compressive sensing, called
  \emph{Info-Greedy Sensing}, which greedily picks the
  measurement with the largest amount of information gain
  based on the previous measurements.
  More
  precisely, we design the next measurement to maximize the
  conditional mutual information between the measurement and the
  signal with respect to the previous
  measurements.
This framework enables us to better understand existing algorithms,
establish theoretical performance guarantees, as well as develop new
algorithms. The optimization problem associated with Info-Greedy
Sensing is often non-convex. In some cases the solutions can be found
analytically, and in others we resort to iterative heuristics. In
particular, (1) we show that the widely used bisection approach is Info-Greedy
  for a family of $k$-sparse signals by connecting compressed sensing
  and blackbox complexity of sequential query algorithms
  \cite{bgp2013}; (2) we present Info-Greedy algorithms for Gaussian and Gaussian
  Mixture Model (GMM) signals under more general noise models
  (e.g. ``noise-folding'' \cite{CastroEldar2011}) than those
  considered in \cite{CarsonChenRodrigues2012}, and analyze their
  performance in terms of the number of measurements needed; (3) we also develop new sensing algorithms, e.g., for sparse sensing
  vectors. 
Numerical examples are provided to demonstrate the accuracy of
theoretical bounds and good performance of Info-Greedy Sensing
algorithms using simulated and real data.  
}

The rest of the paper is organized as follows. Section
\ref{sec:compressed-sensing} sets up the formalism for Info-Greedy
Sensing. Section \ref{sec:k_sparse} and Section \ref{sec:Gaussian}
present the Info-Greedy Sensing algorithms for $k$-sparse signals and
Gaussian signals (low-rank single Gaussian and GMM), respectively.
Section \ref{sec:sparse} discusses the Info-Greedy Sensing with sparse
measurement vectors. Section \ref{sec:egs} contains numerical examples
using simulated and real data. Finally, Section \ref{sec:con}
concludes the paper.  All proofs are delegated to the Appendix.

The notation in this paper is standard. In particular,
$\mathcal{N}(\mu, \Sigma)$ denotes the Gaussian distribution with mean
$\mu$ and covariance matrix $\Sigma$; $[x]_i$ denotes the
$i$th coordinate of the vector $x$; we use the shorthand
\([n] = \{1, \ldots, n\}\);
let \(\card{S}\)
denote the cardinality of a set $S$; $\|x\|_0$ is the number of
non-zeros in vector $x$; let $\|\Sigma\|$ be the spectral norm (largest eigenvalue) of a positive definite matrix $\Sigma$; let \(\det(X)\)
be the determinant of a matrix $X$;
let $\entropy{x}$ denote the entropy of a random variable $x$; let
$\mutualInfo{x}{y}$ denote the mutual information between two random
variables $x$ and $y$.  Let the column vector $e_i$ has \(1\)
on the $i$th entry and zero elsewhere, and let \(\chi_{n}^{2}\)
be the quantile function of the chi-squared distribution with \(n\)
degrees of freedom.

\section{Formulation}
\label{sec:compressed-sensing}

A typical compressed sensing setup is as follows. Let $x \in \mathbb{R}^n$ be the unknown $n$-dimensional signal. There are $m$ measurements, and $y \in \mathbb{R}^m$ is the measurement vector depending linearly on the signal \(x\) and subject to an additive noise: 
\begin{equation}
y = A x + w, \label{one-shot-model}
  \quad
  A \triangleq
  \begin{bmatrix}
    \transpose{a_{1}} \\
    \vdots
    \\
    \transpose{a_{m}}
  \end{bmatrix}
  \in \mathbb{R}^{m\times n},
  \quad 
  w \triangleq 
  \begin{bmatrix}
  w_1 \\ \vdots \\ w_m
  \end{bmatrix} \in \mathbb{R}^{m\times 1},
\end{equation}
where $A \in \mathbb{R}^{m\times n}$ is the {sensing matrix},
and $w \in \mathbb{R}^m$ is the {noise vector}.
Here, each coordinate \(y_{i}\) of $y$ is a result of
measuring $\transpose{a_i} x$ with an additive
noise \(w_{i}\). In the setting of sequential compressed sensing, the unknown signal $x$ is measured
\emph{sequentially}
\[ y_i = \transpose{a_i} x + w_i, \quad i = 1, \ldots, m.\]

In high-dimensional problems, various low-dimensional signal models
for $x$ are in common use: (1) sparse signal models, the canonical one being
  \(x\) having $k \ll n$ non-zero entries\footnote{In a
    related model the signal $x$ come from a dictionary
    with few nonzero coefficients, whose support
    is unknown. We will not further consider this model here.}; (2) low-rank Gaussian model (signal in a
  subspace plus Gaussian noise); and (3) Gaussian mixture model (GMM) (a model for signal lying in a
  union of multiple subspaces plus Gaussian noise), which has been
  widely used in image and video analysis among
  others\footnote{A mixture of GMM models has also been used to study
    sparse signals \cite{MChenThesis}. There are also other
    low-dimensional signal models including the general manifold
    models which will not be considered here.}.

Compressed sensing exploits the low dimensional structure of the
signal to recover the signal with high accuracy using much fewer
measurements than the dimension of the signal, i.e., $m \ll n$.  Two
central and interrelated problems in compressed sensing include signal
recovery and designing the sensing matrix $A$.  Early compressed
sensing works usually assume $A$ to be random, which does have
benefits for universality regardless of the signal
distribution. However, when there is prior knowledge about the signal
distribution, one can optimize \(A\)
to minimize the number \(m\)
of measurements subject to a total sensing power constraint
\begin{equation}
\sum_{i=1}^m \tnorm{a_i}^2 \leq P
\end{equation}
for some constant $P > 0$. In the following, we either vary power for each measurement $\|a_i\|_2^2 = \beta_i$, or fix them to be unit power $\|a_i\|_2 = 1$ (for example, due to physical constraint) and use repeated measurements $\beta_i$ times in the direction of $a_i$, which is equivalent to measuring using an integer valued power. Here $\beta_i$ can be viewed as the amount of resource we allocated for that measurement (or direction).

We will consider a methodology where $A$ is chosen to extract the most
information about the signal, i.e., to maximize
\emph{mutual information}. In the non-sequential setting this means
that $A$ maximizes the mutual information between the signal $x$ and
the measurement outcome, i.e.,
$ A^* = \argmax_A \mutualInfo{x}{Ax + \yx{w}} $.
In sequential compressed sensing, the subsequent measurement vectors
can be designed using the already acquired measurements, and hence the
sensing matrix $A$ can be designed row by row. Optimal sequential
design of \(A\)
can be defined recursively and viewed as dynamic programming
\cite{EldarKutyniok2012}.  However, this formulation is usually
intractable in all but the most simple situations (one such example is
the sequential probabilistic bisection algorithm in
\cite{WaeberFrazier2013}, which locates a single non-zero
entry). Instead, the usual approach operates in a greedy fashion.  The
core idea is that based on the information that the previous
measurements have extracted, the new measurement should probe in the
direction that maximizes the conditional information as much as
possible.
We formalize this idea as \emph{Info-Greedy Sensing}, which is
described in Algorithm \ref{alg:max_mutual_info_1}.
The algorithm is initialized with a prior distribution of signal $x$, and returns the Bayesian posterior mean as an estimator for signal \(x\).  Conditional mutual information is a natural metric, as it counts only useful new information between the signal
and the potential result of the measurement disregarding noise and what has already been learned from previous measurements. 
\begin{algorithm}
\caption{Info-Greedy Sensing}
\begin{algorithmic}[1]
\REQUIRE distributions of signal $x$ and noise \(w\),
  error tolerance $\varepsilon$
  or maximum number of iterations $M$
\STATE \(i \leftarrow 1\)
\REPEAT
\STATE $a_{i} \leftarrow \argmax_{a_i}
  \mutualInfo[y_{j}, a_{j}, j < i]{x}
  {\transpose{a_{i}} x + w_i}/{\beta_i}$
\STATE  \(y_{i} = \transpose{a_{i}} x + w_i\) \COMMENT{measurement}
\STATE $i \leftarrow i + 1$
\UNTIL{$\mutualInfo[y_{j}, a_{j}, j\leq i]{x}
  { \transpose{a_{i}} x + w_i} \leq \delta(\varepsilon)$ or $i > M$.}
\end{algorithmic}
\label{alg:max_mutual_info_1}
\end{algorithm}

\yx{Algorithm \ref{alg:max_mutual_info_1} stops either when the
  conditional mutual information is smaller than a threshold
  $\delta(\varepsilon)$, or we have reached the maximum number of
  iterations $M$. How $\delta(\varepsilon)$ relates to the precision
  $\varepsilon$ depends on the specific signal model employed.  For
  example, for Gaussian signal, the conditional mutual information is
  the log determinant of the conditional covariance matrix, and hence
  the signal is constrained to be in a small ellipsoid with high
  probability.  Also note that in this algorithm, the recovered signal
  may not reach accuracy $\varepsilon$ if it exhausts the number of
  iterations $M$.  In theoretical analysis we assume $M$ is
  sufficiently large to avoid it.}

Note that the optimization problem in
Info-Greedy Sensing
$\argmax_{a_i} \mutualInfo[y_{j}, a_{j}, j < i]{x}
{\transpose{a_{i}} x + w_i}$ 
is non-convex in general
\cite{PayaroPalomar2009}.
\yx{Hence, we will discuss various heuristics and establish their
  theoretical performance in terms of the following metric:}
\begin{defn}[Info-Greedy] We call an algorithm \emph{Info-Greedy} if the
  measurement maximizes
  \(\mutualInfo[y_{j} : j < i]{x}{y_{i}} / \beta_{i}\) \yao{for each $i$},
  where \(x\) is the unknown signal,
  \(y_{j}\) is the measurement outcome,
  and
  \(\beta_{i}\) is the amount of resource for measurement
  \(i\).
\label{def1}
\end{defn}

\section{$k$-sparse signal} \label{sec:k_sparse}

\yx{
In this section, we consider the Info-Greedy Sensing for $k$-sparse
signal with arbitrary nonnegative amplitudes in the noiseless case as
well as under Gaussian measurement noise. We
show that a natural modification of the bisection algorithm corresponds
to Info-Greedy Sensing under a certain probabilistic model. We also show that Algorithm~\ref{alg:bisection} is
optimal in terms of the number of measurements for \(1\)-sparse signals as well as optimal up to a \(\log k\)
factor for \(k\)-sparse signals in the noiseless case. In the presence of Gaussian measurement noise, it is optimal up to at most another \(\log n\) factor. Finally, we show Algorithm~\ref{alg:bisection} is Info-Greedy when \(k=1\), and when \(k > 1\) 
it is \yx{Info-Greedy} up to a \(\log k\) factor.

To simplify the problem, we assume
the sensing matrix $A$ consists of binary entries: \(a_{ij}
\in \binSet\). Consider  a signal with each element $x_i \in \R_+$ with up to
\(k\) non-zero entries which are distributed uniformly at random.
The following lemma gives an upper bound on the number of measurements
$m$ for our modified bisection algorithm (see Algorithm \ref{alg:bisection}) to
recover such \(x\). In the description of Algorithm \ref{alg:bisection}, let
 \[[a_{S}]_{i} \coloneqq
    \begin{cases}
      1, & i \in S \\
      0, & i \notin S
    \end{cases}\] denote the characteristic vector of a set \(S\). 
The basic idea is to recursively estimate a tuple $(S, \ell)$ that
consists of a set $S$ which contains possible locations of the
non-zero elements, and the total signal amplitude in that
set. 
We say that a signal \(x\) has minimum amplitude \(\tau\), if \(x_{i} > 0\)
  implies \(x_{i} \geq \tau\) for all \(i \in [n]\).

\begin{thm}[Upper bound for $k$-sparse signal
$x$]\label{lemma1}
Let \(x \in \mathbb{R}_{+}^n\) be a \(k\)-sparse
signal. 
\begin{enumerate}
\item In the noiseless case, Algorithm \ref{alg:bisection}
recovers the signal \(x\) exactly with
at most \(2 k \lceil \log n \rceil\) measurements
(using \(r = 1\) in Line~\ref{line:r}).
\item In the noisy case with \(w_i\sim \mathcal{N}(0, \sigma^2)\), Algorithm \ref{alg:bisection}
  recovers the signal \(x\) such that
  \(\|x - \widehat{x}\|_2 \leq \sqrt{k}\varepsilon\) with probability at least \(1-k \lceil \log n \rceil / (n^{\varepsilon^{2} / (2 k \sigma^{2})})
  =
  O(1)\)
  using at most
  \(2 k \lceil \log n \rceil^{2}\) measurements.
\end{enumerate}
\end{thm}

\yx{
\begin{algorithm}[h!]
  \caption{Bisection for \(k\)-sparse signals}
  \begin{algorithmic}[1]
    \REQUIRE ambient dimension \(n\) of \(x\), error probability
    \(\delta\), noise variance \(\sigma\), error \(\varepsilon\)
    \STATE \(r \leftarrow \lceil \log n \rceil\)
    \label{line:r}
    \STATE \(L \leftarrow \face{[n]}\)
    \STATE \(\widehat{x} \leftarrow 0\)
    \COMMENT{initialize estimator}
    \WHILE{\(L\) not empty}\label{line:loop}
    \FORALL{\(S \in L\)}
      \STATE Partition \(S = S_1 \dot\cup S_2\) with
      \(\abs*{\card{S_1}-\card{S_2}} \leq 1\)
      \STATE Replace \(S\) by \(S_{1}\) and \(S_{2}\) in \(L\)
    \ENDFOR
    \FORALL{\(S \in L\)}
      \STATE Measure \(r\) times and average:
      \(y = \transpose{a_{S}} x + \yx{w}\)
      \IF{\(y \leq \varepsilon\)}
      \label{line:test}
      \STATE Remove \(S\) from \(L\). \COMMENT{\(\widehat{x}\) is
        already \(0\) on \(S\).}
      \ELSIF{\(\card{S} = 1\)}
      \STATE Remove \(S\) from \(L\).
      \STATE \(\widehat{x}_{i} \leftarrow y\) where \(S = \{i\}\).
      \label{line:estimate}
      \ENDIF
\ENDFOR
    \ENDWHILE
    \RETURN \(\widehat{x}\) as estimator for \(x\).
  \end{algorithmic}
  \label{alg:bisection}
\end{algorithm}
}
%
}

\begin{lem}[Lower bound for noiseless $k$-sparse signal $x$] \label{lem:LB-bisection}
  Let \(x \in \mathbb{R}_+^n\), \( x_i \in \binSet\)
  be a \(k\)-sparse signal.
  Then to recover \(x\) exactly,
  the expected number of measurements $m$ required for any
  algorithm is at least \(\frac{k}{\log k +1} (-1 + \log n)\).
\end{lem}

\begin{lem}[Bisection Algorithm~\ref{alg:bisection} for \(k=1\) is Info-Greedy]\label{lem:infoGreedyBis}
  For \(k=1\) Algorithm~\ref{alg:bisection}, is Info-Greedy.
 \label{lemma_lower_bound} 
\end{lem}
In general case the simple analysis that leads to Lemma \ref{lemma_lower_bound} fails. However, using
Theorem~\ref{thm:estimateRunningTime} in the Appendix we can estimate the average
amount of information obtained from a measurement:
\begin{lem}[Bisection Algorithm~\ref{alg:bisection} is Info-Greedy up
  to a \(\log k\) factor in the noiseless case]\label{lem:bisection-k-ub} 
  Let \(k \leq n \in \N\).
  Then the average information of a measurement
  in Algorithm~\ref{alg:bisection}:
\[\mutualInfo[Y_1,\dots, Y_{i-1}]{X}{Y_i} \geq 1-
\frac{\log k}{\log n}.\]
\end{lem}

\begin{rem}
\begin{enumerate}
\item \yx{Observe that Lemma~\ref{lem:bisection-k-ub} establishes that
  Algorithm~\ref{alg:bisection} for a sparse signal with
  \(\log k = o(\log n)\)
  acquires at least a \(\frac{1}{\log k+1} - o(1)\)
  fraction of the maximum possible mutual information (which on
  average  is roughly \(1\) bit per measurement).}
  \item \yx{Here we constrained the entries of matrix $A$ to be binary valued. This may correspond to applications, for examples, sensors reporting errors and the measurements count the total number of errors. Note that, however, if we relax this constraint and allow entries of $A$ to be real-valued, in the absence of noise the signal can be recovered from one measurement that project the signal onto a vector with entries $[2^0, 2^1, 2^2, \cdots]$.}
\item \yx{The setup here with $k$-sparse signals and binary measurement
matrix $A$ generalizes the group testing \cite{IwenTewfikGroup2010} setup.}
\item the CASS algorithm \cite{MalloyNowak2013} is another algorithm
  that recovers a $k$-sparse signal $x$ by iteratively partitioning
  the signal support into $2k$ subsets, computing the sum over that
  subset and keeping the largest $k$. In \cite{MalloyNowak2013} it was
  shown that to recover a $k$-sparse \(x\)
  with non-uniform positive amplitude with high probability, the
  number of measurements $m$ is on the order of $2k\log (n/k)$ with {\it varying} power measurement. 
   It is important to note that the CASS algorithm
  allows for power allocation to mitigate noise, while we repeat
  measurements. This, however, coincides with the number of unit length measurements of our algorithm,  $2k \lceil \log n\rceil^2$ in Lemma \ref{lemma1} after
  appropriate normalization.
 %
  For specific regimes of error probability, the \(O(\log n)\)
  overhead in Lemma \ref{lemma1} can be further reduced. For example,
  for any constant probability of error \(\varepsilon > 0\),
  the number of required repetitions per measurement is \(O(\log \log n
  )\) leading to improved performance. Our algorithm can be also
  easily modified  to incorporate power allocation.
\end{enumerate}
\end{rem}

\section{Low-Rank Gaussian Models} \label{sec:Gaussian}

In this section, we derive the Info-Greedy Sensing algorithms for the
single low-rank Gaussian model as well as the low-rank GMM signal
model, and also quantify the algorithm's performance. 

\subsection{Single Gaussian model}

Consider a Gaussian signal \(x \sim \Dnormal{\mu}{\Sigma}\)
with known parameters $\mu$ and $\Sigma$. The covariance matrix $\Sigma$ has rank $k \leq n$. 
We will consider three noise models:
\begin{enumerate}
\item white Gaussian noise added after the measurement (the most common model
  in compressed sensing):
\begin{equation}
y = Ax + w, \quad w\sim \mathcal{N}(0, \sigma^2 I).
\end{equation} 
Let $\beta_i = \|a_i\|_2^2$ represent the power allocated to the $i$th measurement. In this case, higher power $\beta_i$ allocated to a measurement increases SNR of that measurement. 
\item white Gaussian noise added prior to the measurement, a model that appears in
  some applications such as reduced dimension multi-user detection in
  communication systems \cite{XieEldarGoldsmith2013} and also known as
  the ``noise folding'' model \cite{CastroEldar2011}:
\begin{equation}
y = A(x + w), w \sim \mathcal{N}(0, \sigma^2 I).
\end{equation}
In this case, allocating higher power for a measurement cannot increase the SNR of the outcome. Hence, we use the actual number of repeated measurements in the same direction as a proxy for the amount of resource allocated for that direction.
\item colored Gaussian noise with covariance $\Sigma_w$ added
  either prior to the measurement:
\begin{equation}
y = A(x + w), w \sim \mathcal{N}(0, \Sigma_w),
\end{equation}
or after the measurement:
\begin{equation}
y = Ax + w, w \sim \mathcal{N}(0, \Sigma_w).
\end{equation}
\end{enumerate}

In the following, we will establish lower bounds on the amount of resource (either the minimum power or the number of measurements) needed for Info-Greedy Sensing to achieve a recovery error \(\|x - \widehat{x}\|_2\leq \varepsilon\).

\subsubsection{White noise added prior to measurement or ``noise folding''}
\label{sec:isotrope-noise}

We start our discussion with this model and results for other models can be derived similarly. As  \(\beta_{i}\)
does not affect SNR, we set $\|a_i\|_2 = 1$. 
Note that conditional distribution of \(x\) given \(y_{1}\) is a Gaussian random vector with adjusted parameters
\begin{equation}
  \label{eq:3}
  \begin{split}
  & x \mid y_{1} \sim \mathcal{N}(\mu
    + \Sigma a_{1}
    (\transpose{a_{1}} \Sigma a_{1} + \sigma^{2})^{-1}
    (y_{1} - \transpose{a_{1}} \mu), \\
  &~~\Sigma - \Sigma a_{1}
    (\transpose{a_{1}} \Sigma a_{1} + \sigma^{2})^{-1}
    \transpose{a_{1}} \Sigma).
    \end{split}
\end{equation}
Therefore, to find Info-Greedy Sensing for a single Gaussian signal, it suffices to characterize the first measurement $a_1 = \arg\max_{a_1} \mutualInfo{x}{y_{1}}$ and from there on iterate with adjusted distributional parameters. 
For Gaussian signal $x \sim \mathcal{N}(\mu, \Sigma)$
and the noisy measurement $\sigma > 0$, we have 
\begin{equation}
  \label{eq:Gaussian-noise-mutual}
  \mutualInfo{x}{y_{1}} =
  \entropy{y_{1}} - \entropy[x]{y_1}
  = \frac{1}{2}\ln
  \left(
    \transpose{a_{1}} \Sigma a_{1}/\sigma^{2} + 1
  \right). 
\end{equation}
Clearly, with $\|a_1\|_2 = 1$, (\ref{eq:Gaussian-noise-mutual}) is maximized when $a_1$ corresponds to the largest \yx{eigenvector} of \(\Sigma\). 
From the above argument, the Info-Greedy Sensing algorithm for a single Gaussian signal is to choose \(a_{1}, a_{2}, \dotsc\) as the
orthonormal eigenvectors of \(\Sigma\) in a decreasing order of eigenvalues, as described in Algorithm~\ref{alg:Gaussian-all}. 
The following theorem establishes the bound on the number of measurements needed.
\begin{thm}[White Gaussian noise added prior to measurement or ``noise folding'']
  \label{thm:Gaussian}
  Let \(x \sim \Dnormal{\mu}{\Sigma}\) and let
  \(\lambda_{1}, \dotsc, \lambda_{k}\) be the eigenvalues
  of \(\Sigma\) with multiplicities.
  Further let \(\varepsilon > 0\) be the
  accuracy and  \(w_i \sim \Dnormal{0}{\sigma^{2}}\). 
  Then Algorithm~\ref{alg:Gaussian-all} recovers
  \( x\)
  satisfying \(\tnorm{x-\widehat x} < \varepsilon\)  with probability
  at least
  \(p\) using at most the following number of measurements
  by unit vectors \(\|a_i\|_2 = 1\):
  \begin{subequations}
    \begin{equation}
      \label{eq:Gaussian-number-noisy}
      m = \sum_{\substack{i=1 \\ \lambda_{i} \neq 0}}^{k}
      \max\left\{
        0, \ceil*{
          \left(
            \frac{\chi_{n}^{2}(p)}{\varepsilon^{2}}
            - \frac{1}{\lambda_{i}}
          \right)
          \sigma^{2}}
      \right\}
    \end{equation}
    provided \(\sigma > 0\). If \(\sigma^{2} \leq \varepsilon^{2} /
    \chi_{n}^{2}(p)\) the number of measurements simplifies to
    \begin{equation}
      \label{eq:Gaussian-number-exact}
      \card{\left\{
        i : \lambda_{i} > \frac{\varepsilon^{2}}{\chi_{n}^{2}(p)}
      \right\}}.
    \end{equation}
    This also holds when \(\sigma = 0\).
  \end{subequations}
\end{thm}

\subsubsection{White noise added after measurement}

A key insight in the proof for Theorem \ref{thm:Gaussian} is that
repeated measurements in the same eigenvector direction
corresponds to a single measurement in that direction with all the power summed together. 
\yx{This can be seen from the following discussion. After measuring in the direction of a unit norm eigenvector $u$ with eigenvalue \(\lambda\),
and using power \(\beta\),
the conditional covariance matrix takes the form of}
\begin{equation}
  \label{eq:covariance-eigenvector}
  \begin{split}
  &\Sigma - \Sigma \sqrt{\beta} u
  \left(
    \transpose{\sqrt{\beta} u} \Sigma \sqrt{\beta} u
    + \sigma^{2}
  \right)^{-1}
  \transpose{\sqrt{\beta} u} \Sigma \\
&  =
  \frac{\lambda \sigma^{2}}{\beta \lambda + \sigma^{2}}
  u \transpose{u}
  + \Sigma^{\perp u},
  \end{split} 
\end{equation}
where \(\Sigma^{\perp u}\) is the component of \(\Sigma\)
in the orthogonal complement of \(u\).
Thus, the only change in the eigendecomposition of \(\Sigma\)
is the update of the eigenvalue of \(u\)
from \(\lambda\) to
\(\lambda \sigma^{2} / (\beta \lambda + \sigma^{2})\).
\yx{Informally, measuring with power allocation $\beta$ on a Gaussian signal \(x\)  reduces the uncertainty in direction \(u\) as illustrated in Fig. \ref{fig:Gaussian_mat}. We have the following performance bound for sensing a Gaussian signal:}
%
%
\begin{thm}[White Gaussian noise added after measurement]
  \label{thm:Gaussian-fixed-noise}
  Let \(x \sim \Dnormal{\mu}{\Sigma}\) and let
  \(\lambda_{1}, \dotsc, \lambda_{k}\) be the eigenvalues
  of \(\Sigma\) with multiplicities.
  Further let \(\varepsilon > 0\) be the
  accuracy and   \(w_i\sim \Dnormal{0}{\sigma^{2}}\).
  Then Algorithm~\ref{alg:Gaussian-all}
  recovers \( x\) satisfying \(\tnorm{x-\widehat x} < \varepsilon\)  with probability
  at least
  \(p\) using at most the following power 
  \begin{equation}
    \label{eq:Gaussian-number-fix-noise}
    P = \sum_{\substack{i=1 \\ \lambda_{i} \neq 0}}^{k}
    \max\left\{
      0,
      \left(
        \frac{\chi_{n}^{2}(p)}{\varepsilon^{2}}
        - \frac{1}{\lambda_{i}}
      \right)
      \sigma^{2}
    \right\}
  \end{equation}
  provided \(\sigma > 0\). 
\end{thm}

\begin{figure}[h!]
\begin{center}
  \begin{tikzpicture}[every node/.style={outer sep=.5ex,
      drop shadow=black},
    line width=2pt]
    \node[ellipse, fill=green!11!yellow!66!cyan!75!black,
    minimum width=7em, minimum height=6em](orig){\(\Sigma\)};
    \node[ellipse, fill=green!68!cyan!53!black,
    minimum width=6em, minimum height=4em,
    right=of orig](first){\(\Sigma_{x \mid y_{1}}\)};
    \node[ellipse, fill=green!84!yellow!74!cyan!36!black,
    minimum width=1em, minimum height=.5em,
    right=of first](second)
    {\(\Sigma_{x \mid y_{1}, y_{2}}\)};
    \coordinate[right=of second](dummy);

    \draw[->] (orig) -- (first);
    \draw[->] (first) -- (second);
  \end{tikzpicture}
\caption{Evolution of the covariance matrix by
  sequentially measuring with
  an eigenvector of the largest eigenvalue.}
  \label{fig:Gaussian_mat}
\end{center}
\end{figure}
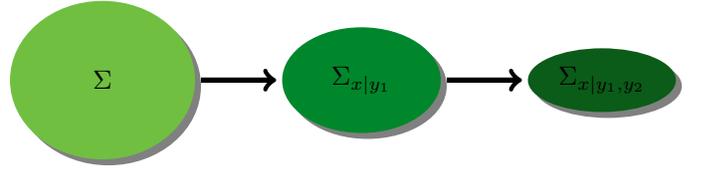

\subsubsection{Colored noise}
\label{sec:general-noise}

When a colored noise $w\sim \mathcal{N}(0, \Sigma_w)$ is added either prior to, or after the measurement, similar to the white noise cases, the conditional distribution of \(x\) given the first measurement \(y_1\) is a Gaussian random variable with adjusted parameters. Hence, as before, the measurement vectors can be found iteratively. 
Algorithm~\ref{alg:Gaussian-all} presents Info-Greedy Sensing for this case and the derivation is given in Appendix \ref{app:Gaussian_color}.  Algorithm~\ref{alg:Gaussian-all} also summarizes all the Info-Greedy Sensing algorithms for Gaussian signal under various noise models.


%

\begin{algorithm}[h!]
  \caption{Info-Greedy Sensing for Gaussian signals}
  \begin{algorithmic}[1]
    \REQUIRE signal mean \(\mu\) and covariance \(\Sigma\),
      accuracy \(\varepsilon\),
      probability of correctness \(p\), noise covariance matrix $\Sigma_w$ (for white noise \(\sigma^2I\) )
    \REPEAT 
       \IF{white noise added after measurement}
             \STATE \(\lambda \leftarrow \norm{\Sigma}\)
        \COMMENT{largest eigenvalue}
            \STATE \(u \leftarrow\)
        eigenvector of \(\Sigma\)
        for eigenvalue \(\lambda\)
  \STATE
        \(\beta
        \leftarrow
          \left(
            \frac{\chi_{n}^{2}(p)}{\varepsilon^{2}}
            - \frac{1}{\lambda}
          \right)
          \sigma^{2}
        \)
        \STATE \(a\leftarrow \sqrt{\beta} u\)
        \STATE \(y = \transpose{a} x + w \) 
       \ELSIF{white noise added prior to measurement}
             \STATE \(\lambda \leftarrow \norm{\Sigma}\)
        \COMMENT{largest eigenvalue}
       \STATE \(u \leftarrow\)
        eigenvector of \(\Sigma\)
        for eigenvalue \(\lambda\)
        \STATE \(a\leftarrow \sqrt{\beta} u\)
        \STATE \(y = \transpose{a} (x + w)\)
        \ELSIF{colored noise added after measurement}
        \STATE 
        \hspace{-0.1in}$\Sigma = U_x \Lambda_x \transpose U_x, \Sigma_w = U_w \Lambda_w \transpose U _w$
        \COMMENT{eigendecomposition} \vspace{-0.15in}
        \STATE \(u \leftarrow (1/\|\Lambda_w^{1/2} \transpose U_w e_1\|_2) U_x \Lambda_w^{1/2} \transpose U_w e_1\)
        \STATE \(a\leftarrow \sqrt{\beta} u\)
        \STATE \(y = \transpose{a} x + w \) 
        \ELSIF{colored noise added prior to measurement}
        \STATE \(\lambda \leftarrow \norm{\Sigma_{w}^{-1} \Sigma}\)
        \COMMENT{largest eigenvalue}
      \STATE
        \(\beta
        \leftarrow
        \frac{\chi_{n}^{2}(p)}{\varepsilon^{2}}
        \norm{\Sigma_{w}}
        - \frac{1}{\lambda}
        \)
      \STATE \(u \leftarrow\)
        largest eigenvector of \(\Sigma_{w}^{-1} \Sigma\) for eigenvalue $\lambda$
        \STATE \(a\leftarrow \sqrt{\beta} u\)
      \STATE \(y = \transpose{a} (x + w)\) 
       \ENDIF
       \STATE \(\mu \leftarrow \mu
        + \Sigma a
        (\transpose{a} \Sigma a + \sigma^{2})^{-1}
        (y - \transpose{a}\mu )\) \COMMENT{mean}
        \STATE \(\Sigma\leftarrow \Sigma
        - \Sigma a
        (\transpose{a} \Sigma a + \sigma^{2})^{-1}
        \transpose{a}\Sigma\) \COMMENT{covariance}
           \UNTIL{\(\norm{\Sigma} \leq \varepsilon^{2} / \chi_{n}^{2}(p)\)}
    \COMMENT{all eigenvalues become small}
    \RETURN posterior mean \(\mu\)
  \end{algorithmic}
  \label{alg:Gaussian-all}
\end{algorithm}

The following version of Theorem~\ref{thm:Gaussian} is for the required number of measurements for colored noise in the ``noise folding'' model: 
\begin{thm}[Colored Gaussian noise added prior to measurement or ``noise folding'']
  \label{thm:Gaussian-non-iso}
  Let \(x \sim \Dnormal{\mu}{\Sigma}\) be a Gaussian signal,
  and let
  \(\lambda_{1}, \dotsc, \lambda_{n}\) denote the eigenvalues
  of \(\Sigma_{w}^{-1} \Sigma\) with multiplicities.
  Assume 
  \(w\sim\Dnormal{0}{\Sigma_{w}}\).
  Furthermore, let \(\varepsilon > 0\) be the required accuracy. 
Then Algorithm~\ref{alg:Gaussian-all} recovers \({x}\) 
  satisfying \(\tnorm{x-\widehat{x}} < \varepsilon\)
  with probability at least \(p\) using at most the following number of measurements by unit vectors $\|a_i\|_2 = 1$: 
    \begin{equation}
      \label{eq:Gaussian-number-non-iso}
      m = \sum_{\substack{i=1 \\ \lambda_{i} \neq 0}}^{n}
      \max\left\{
        0,
          \left \lceil
            \frac{\chi_{n}^{2}(p)}{\varepsilon^{2}}
            \norm{\Sigma_{w}}
            - \frac{1}{\lambda_{i}}
          \right \rceil
      \right\}.
    \end{equation}
\end{thm}

\begin{rem}
(1) Under these noise models, the posterior distribution of the signal is also Gaussian, and the measurement outcome \(y_i\) affects \emph{only its mean} and but not the covariance matrix (see
  \eqref{eq:3}). In other words, the outcome does not affect the mutual information of posterior Gaussian signal. In this sense, for Gaussian signals adaptivity brings no advantage {\it when $\Sigma$ is accurate}, as the measurements are pre-determined by the eigenspace of $\Sigma$. \yx{However, when knowledge of $\Sigma$ is inaccurate for Gaussian signals, adaptivity brings benefit as demonstrated in Section \ref{sec:eg_Gaussian}, since a sequential update of the covariance matrix incorporates new information and ``corrects'' the covariance matrix when we design the next measurement.}

(2) In (\ref{eq:covariance-eigenvector}) the eigenvalue \(\lambda\) reduces to
\(\lambda \sigma^{2} / (\beta \lambda + \sigma^{2})\)  after the first measurement. Now iterating
this we see by induction that after \(\yx{m'}\) measurements in direction
\(a\), the eigenvalue \(\lambda\) reduces to 
\(\lambda \sigma^{2} / (m' \beta \lambda + \sigma^{2})\), which is the
same as measuring once in direction \(a\) with power \(m' \beta\). Hence, measuring several times in the same direction of \(a\), and thereby splitting power
into \(\beta_{1}, \dotsc, \beta_{m'}\) for the measurements,
has the same effect as making one measurement
with total the power \(\sum_{i=1}^{m'} \beta_{i}\).

(3) Info-Greedy Sensing for Gaussian signal can be
  implemented efficiently. \yx{Note that in the algorithm we only need compute the leading eigenvector of the covariance matrix; moreover, updates of the covariance matrix and mean are  simple and iterative.} In particular, for a sparse
  $\Sigma \in \mathbb{R}^{n\times n}$ with $v$ non-zero entries, the
  computation of the largest eigenvalue and associated eigenvector can
  be implemented in $\mathcal{O}(t(n+v))$ using sparse power's method
  \cite{sparseEigenvector2011}, where $t$ is the number of power
  iterations.   
  In many high-dimensional applications, $\Sigma$ is sparse if the
  variables (entries of $x$) are not highly correlated.  Also note that
    the sparsity structure of the covariance matrix as well as the
    correlation structure of the signal entries will not be changed by
    the update of the covariance matrix. This is because
    in (\ref{eq:covariance-eigenvector}) the update only changes the
    eigenvalues but not the eigenvectors. To see why this is true, let
    $\Sigma = \sum_i \lambda_i q_i q_i^\top$ be the eigendecomposition of $\Sigma$. By saying that the
    covariance matrix is sparse, we assume that $q_i$'s are sparse and,
    hence, the resulting covariance matrix $\Sigma$ has few 
    non-zero entries. Therefore, updating the covariance
    matrix will not significantly change the
    number of non-zero entries in a covariance matrix. We demonstrate the scalability of Info-Greedy Sensing with larger examples in Section \ref{sec:eg_Gaussian}.
  %
\end{rem}


\subsection{Gaussian mixture model (GMM)}

The probability density function of GMM
is given by 
\begin{equation}
p(x) = \sum_{c=1}^C \pi_c \mathcal{N}(\mu_c, \Sigma_c), \label{GMM_model}
\end{equation}
where $C$ is the number of classes,
and $\pi_c$ is the probability of samples from class $c$. 
Unlike Gaussian, mutual information of GMM cannot be explicitly written. However, for GMM signals a gradient descent approach that works for an arbitrary signal model can be used as outlined in \cite{CarsonChenRodrigues2012}. The derivation uses the fact that the gradient of the conditional mutual information with respect to $a_i$ is a linear transform of the minimum mean square error (MMSE)  matrix \cite{PalomarVerdu2006, PayaroPalomar2009}. 
\yx{Moreover, the gradient descent approach for GMM signals exhibits structural properties that can be exploited to reduce the computational cost for evaluating the MMSE matrix, as outlined in \cite{MChenThesis,CarsonChenRodrigues2012}. 
For completeness we include the detail of the algorithm here, as summarized in Algorithm \ref{alg:GMM_heuristics} and the derivations are given in Appendix \ref{app:GMM}\footnote{\yx{Another related work is \cite{RennaCalderbank2014} which studies the behavior of minimum mean sure error (MMSE) associated with the reconstruction of a signal drawn from a GMM as a function of the properties of the linear measurement kernel and the Gaussian mixture, i.e. whether the MMSE converges or does not converge to zero as the noise. }}.}


An alternative heuristic for sensing GMM is the so-called \emph{greedy heuristic}, \yx{which is also mentioned in \cite{CarsonChenRodrigues2012}}. The heuristic picks the Gaussian component with the highest posterior $\pi_c$ at that moment, and chooses the next measurement $a$ to be its eigenvector associated with the maximum eigenvalue, as summarized in Algorithm \ref{alg:GMM_heuristics}. The greedy heuristic is not Info-Greedy, but it can be implemented more efficiently compared to the gradient descent approach. 
The following theorem establishes a simple upper bound on the number of required measurements to recover a GMM signal using the greedy heuristic with small error. 
The analysis is based on the well-known
multiplicative weight update method (see e.g.,
\cite{arora2012multiplicative}) and utilizes a simple
reduction argument showing that when the variance of every component has been reduced sufficiently to ensure a low error recovery with probability \(p\), we can learn (a mix of) the right component(s) with few extra
measurements. 

\begin{thm}[Upper bound on $m$ of greedy heuristic algorithm for GMM]\label{thm:GMM}
Consider a GMM signal $x$ parameterized in (\ref{GMM_model}). Let \(m_c\) be the required number of
measurements (or power) to ensure \(\tnorm{x - \widehat x} < \varepsilon\) 
with probability \(p\) for a Gaussian signal
\(\Dnormal{\mu_c}{\Sigma_c}\) corresponding to component \(c\) for all \(c \in
C\). Then we need at most
\[\left(\sum_{c \in C} m_c \right ) + \Theta(\frac{1}{\tilde{\eta}} \ln \card{C})\]
measurements  (or power) to ensure \(\tnorm{x - \widehat x} < \varepsilon\) when sampling from
the posterior distribution of \(\pi\) with probability  \(p (1- \tilde{\eta} - o(1))\).
\end{thm}

\begin{rem}
\yx{In the high noise case, i.e., when SNR is low, Info-Greedy measurements can be approximated easily. Let $c_0$ denote the random variable indicating the class where the signal is sampled from. Then $\mutualInfo{x}{y} = \mutualInfo[c]{x}{y} + \mutualInfo{x}{c} - \underbrace{\mutualInfo[y]{x}{c}}_{0} =\mutualInfo{x}{c} + \sum_c \pi_c\log(1+\transpose{a} \Sigma_c a/\sigma^2 )/2\propto \sum_c \pi_c \transpose{a} \Sigma_c a /\sigma^2  = \transpose{a} (\sum_c \pi_c \Sigma_c)a/\sigma^2$. Hence, the Info-Greedy measurement should be the leading eigenvector of the average covariance matrix with the posterior weights.}
\end{rem}

\begin{algorithm}
\caption{Gradient descent for mutual information maximizing measurement}
\begin{algorithmic}[1]
\REQUIRE initial $a_i$, step-size $\mu$, tolerance $\eta > 0$
\REPEAT
\STATE generate $c_0 \sim \widetilde{\pi}_c$, and $x_0 \sim \mathcal{N}(\mu_{c_0} + \Sigma_c D_{i-1}(y_{-i} - D_{i-1} \mu_c)/\sigma^2, \Sigma_{c}
  - \Sigma_{c} \transpose{D_{i-1}} D_{i-1} \Sigma_{c} / \sigma^{2})$.
\STATE measure $y_0 = \transpose{a_{i}} x_0 + w_i$
\STATE evaluate $g(y_0)$ using \eqref{g_def}
 \STATE estimate $E_i \approx \frac{1}{N} \sum_{j=1}^N \tilde{p}(y_j) g(y_j)$. 
 \STATE evaluate $h_i(a_i) \triangleq \partial \mutualInfo[y_j, j < i]{x}{y_i}/\partial a_i$ using  \eqref{gradient}
 \STATE update $a_i \leftarrow a_i + \mu h_i(a_i)$
 \STATE evaluate approximated mutual information using (\ref{MI})
   \UNTIL{increase in mutual information $\leq \eta$}
    \RETURN solution measurement vector $a_i$
\end{algorithmic}
\label{alg:a_grad}
\end{algorithm}

\begin{algorithm}
\caption{Update GMM distributional parameters}
\begin{algorithmic}[1]
\REQUIRE mean \(\{\mu_c\}\), covariance \(\{\Sigma_c\}\), number of GMM components $C$, distribution $\{\pi_c\}$, standard deviation \(\sigma\) of noise, matrix contains vectors thus far $D$ and measurements acquired thus far $\widetilde{y}$ 
\FOR{$c = 1, \ldots, C$}
        \STATE \(\mu_c \leftarrow \mu_c
        + \Sigma_c a
        (\transpose{a} \Sigma_c a + \sigma^{2})^{-1}
        (y - \transpose{a}\mu_c )\) \COMMENT{mean}
        \STATE \(\Sigma_c \leftarrow \Sigma_c
        - \Sigma_c a
        (\transpose{a} \Sigma_c a + \sigma^{2})^{-1}
        \transpose{a}\Sigma_c\) \COMMENT{covariance}
        \STATE \(\pi_c \leftarrow\pi_c \mathcal{N}(\widetilde{y}, D \mu_c 
      D\Sigma_c  \transpose{D} + \sigma^2)\)     
        \ENDFOR
    \STATE  \(\pi_c \leftarrow \pi_c/\sum_{c=1}^C \pi_c\) \COMMENT{normalizing distribution}
\RETURN updated parameters $\{\mu_c, \Sigma_c, \pi_c\}$
\end{algorithmic}
\label{alg:GMM_update}
\end{algorithm}

\begin{algorithm}
  \caption{Info-Greedy Sensing for GMM using greedy heuristic and gradient descent approach
  }
  \begin{algorithmic}[1]
    \REQUIRE mean \(\{\mu_c\}_{c=1}^{C}\), covariance \(\{\Sigma_c\}_{c=1}^C \), initial distribution \(\{\pi_c\}_{c=1}^{C}\)
      standard deviation \(\sigma\) of noise,
      probability of correctness \(p\)
      \STATE Initialize \(\mu^{(0)}_c = \mu_c\), \(\Sigma^{(0)}_c = \Sigma_c\), \(\pi_c^{(0)} = \pi_c\)
    \REPEAT 
                    \IF{greedy heuristic}
             
            \STATE \(z \leftarrow \arg\max_c \pi_c^{(i-1)} \)  
        
             \STATE \(a_{i} \leftarrow\)
        largest eigenvector of \(\Sigma_{z}^{(i-1)}\)
        \ELSIF{gradient decent approach}
         \STATE \(a_{i} \leftarrow\) solved from Algorithm \ref{alg:a_grad}
        \ENDIF
        
          \STATE \(y_i = \transpose{a_i} x + w_i \) \COMMENT{measure}
          
          \STATE update parameters \(\mu_c^{(i)}, \Sigma_c^{(i)}, \pi_c^{(i)}\) using Algorithm \ref{alg:GMM_update} 
  
    \UNTIL{reach maximum iteration 
    }
    \RETURN signal estimate $c^* = \arg\max_c \pi_c^{(I)}$, \(\widehat{\mu} = \mu^{(I)}_{c^*}\)
  \end{algorithmic}
  \label{alg:GMM_heuristics}
\end{algorithm}

\section{Sparse measurement vector} \label{sec:sparse}

In various applications, we are interested in finding a sparse measurement vector $a$. With such requirement, we can add a cardinality constraint on $a$ in the Info-Greedy Sensing formulation: $\|a\|_0 \leq k_0$, where $k_0$ is the number of non-zero entries we allowed for $a$ vector. This is a non-convex integer program with non-linear cost function, which can be solved by outer approximation \cite{DuranGrossmannOuterApprox1986, schrijver1986theory}. The idea of outer approximation is to generate a sequence of cutting planes to approximate the cost function via its subgradient and iteratively include these cutting planes as constraints in the original optimization problem. In particular, we initialize by solving the following optimization problem
\begin{equation}
\begin{array}{rl}
\underset{a, r, z}{\mbox{maximize}} & z \\
\mbox{subject to} & \sum_{i=1}^n r_i \leq k_0 \\
& a_i \leq r_i, \quad -a_i \leq r_i \\
& 0 \leq z \leq c, \quad r_i \in \{0, 1\}, i = 1,\ldots,n \\
& a \in \mathbb{R}^n, \quad z \in \mathbb{R},
\end{array}
\label{sparse_a_solution}
\end{equation}
where $r$ and $z$ are introduced auxiliary variables, and $c$ is an user specified upper bound that bounds the cost function over the feasible region. The constraint of the above optimization problem can be casted into matrix vector form as follows:
\[
F_0 \triangleq  \left[
\begin{array}{c|c|c}
  1_{1\times n} & 0_{1\times n}  & 0\\ \hline
  -I_{n}        & I_n   & 0_{n\times 1} \\ \hline
   -I_{n}        & -I_n   & 0_{n\times 1}\\ \hline
   0_{1\times n} & 0_{1\times n} & 1 \\ \hline
   0_{1\times n} & 0_{1\times n}  & -1 
\end{array}
\right], \quad
g_0 \triangleq \begin{bmatrix}
k_0 \\ 0_{2n\times 1} \\ c \\ 0
\end{bmatrix}
\]
such that
$
F_0 \transpose{\begin{bmatrix}
r & a & z
\end{bmatrix}}
 \leq g_0.
$
The mixed-integer linear program formulated in (\ref{sparse_a_solution}) can be solved efficiently by a standard software such as GUROBI\footnote{http://www.gurobi.com}.  
In the next iteration, solution $a_*$ to this optimization problem will be used to generate a new cutting plane, which we include in the original problem by appending a row to $F_\ell$ and adding an entry to $g_\ell$ as follows
\begin{align}
F_{\ell + 1} & = \left[\begin{array}{c}
F_\ell \\\hline 0 \quad -\transpose{(\nabla f(a_*))} \quad 1
\end{array}\right],\\
g_{\ell + 1} &=\left[\begin{array}{c}
g_\ell \\\hline f(a_*) - \transpose{a_*} \nabla f(a_*)
\end{array}\right],
\end{align}
where $f$ is the non-linear cost function in the original problem. For Gaussian signal $x$, the cost function and its gradient take the form of:
\begin{align}
f(a) = \frac{1}{2} \log(\frac{\transpose a \Sigma a}{\sigma^2} + 1),
\quad  \nabla f(a) = \frac{1}{\transpose a \Sigma a + \sigma^2}\Sigma a.
\end{align}
By repeating iterations as above, we can find a measurement vector with sparsity $k_0$ which is approximately Info-Greedy.

\section{Numerical examples} \label{sec:egs}

\subsection{Simulated examples}

\subsubsection{Low-rank Gaussian model} \label{sec:eg_Gaussian} 
First, we examine the performance of Info-Greedy Sensing for Gaussian signal. The dimension of the signal is $n = 100$, and we set the probability of recovery $p = 0.95$, the noise standard deviation $\sigma = 0.01$. The signal mean vector $\mu = 0$, where the covariance matrix $\Sigma$ is  generated as $\Sigma = \mathcal{T}_{0.7}(\Sigma_0 \transpose\Sigma_0/\|\Sigma_0 \transpose\Sigma_0\|_2)$, $\Sigma_0 \in \mathbb{R}^{n\times n}$ has each entry i.i.d. $\mathcal{N}(0, 1)$, and the operator $\mathcal{T}_{0.7}(X)$ thresholds eigenvalues of a matrix $X$ that are smaller than 0.7 to be zero.  The error tolerance $\epsilon = 0.1$ (represented as dashed lines in the figures). For the white noise case, we set $w \sim \mathcal{N}(0, \sigma^2 I)$, and for the colored noise case, $w\sim \mathcal{N}(0, \Sigma_w)$ and the noise covariance matrix $\Sigma_w$ is generated randomly as $\transpose{\tilde{\Sigma}_0} \tilde{\Sigma}_0/\| \transpose{\tilde{\Sigma}}_0 \tilde{\Sigma}_0\|_2$ for a random matrix $\tilde{\Sigma}_0$ with entries i.i.d. $\mathcal{N}(0, 1)$. The number of measurements is determined from Theorem \ref{thm:Gaussian} and Theorem \ref{thm:Gaussian-fixed-noise}. We run the algorithm over 1000 random instances. Fig. \ref{Fig:comp_white_color_Gaussian} demonstrates the ordered recovery error $\|x - \widehat x\|_2$, as well as the ordered number of measurements calculated from the formulas, for the white and colored noise cases, respectively. Note that in both the white noise and colored noise cases, the errors for Info-Greedy Sensing can be two orders of magnitude lower than the errors obtained from measurement using Gaussian random vectors, and the errors fall below our desired tolerance $\varepsilon$ using the theoretically calculated $m$.

When the assumed covariance matrix for the signal $x$ is equal to its true covariance matrix, Info-Greedy Sensing is identical to the batch method \cite{CarsonChenRodrigues2012} (the batch method measures using the largest eigenvectors of the signal covariance matrix). However, when there is a mismatch between the two, Info-Greedy Sensing outperforms the batch method due to adaptivity, as shown in Fig. \ref{Fig:mismatch}. \yao{For Gaussian signals, the complexity of the batch method is $\mathcal{O}(n^3)$ (due to eigendecomposition), versus the complexity of Info-Greedy Sensing algorithm is on the order of $\mathcal{O}(t m n^2)$ where $t$ is the number of iterations needed to compute the eigenvector associated with the largest eigenvalue (e.g., using the power method),  and $m$ is the number of measures which is typically on the order of $k$. }

We also try larger examples. Fig. \ref{Fig:Gaussian_n1000} demonstrates the performance of Info-Greedy Sensing for a signal $x$ of dimension 1000 and with dense and low-rank $\Sigma$ (approximately $5\%$ of non-zero eigenvalues). Another interesting case is shown in Fig. \ref{Fig:GaussianLarge}, where $n = 5000$ and $\Sigma$ is rank 3 and very sparse: only about $0.0003\%$ of the entries of $\Sigma$ are non-zeros. In this case Info-Greedy Sensing is able to recover the signal with a high precision using only $3$ measurements. This shows the potential value of Info-Greedy Sensing for big data.

\begin{figure}[h!]
\begin{center}
\begin{tabular}{cc}
\includegraphics[width = 0.46\linewidth]{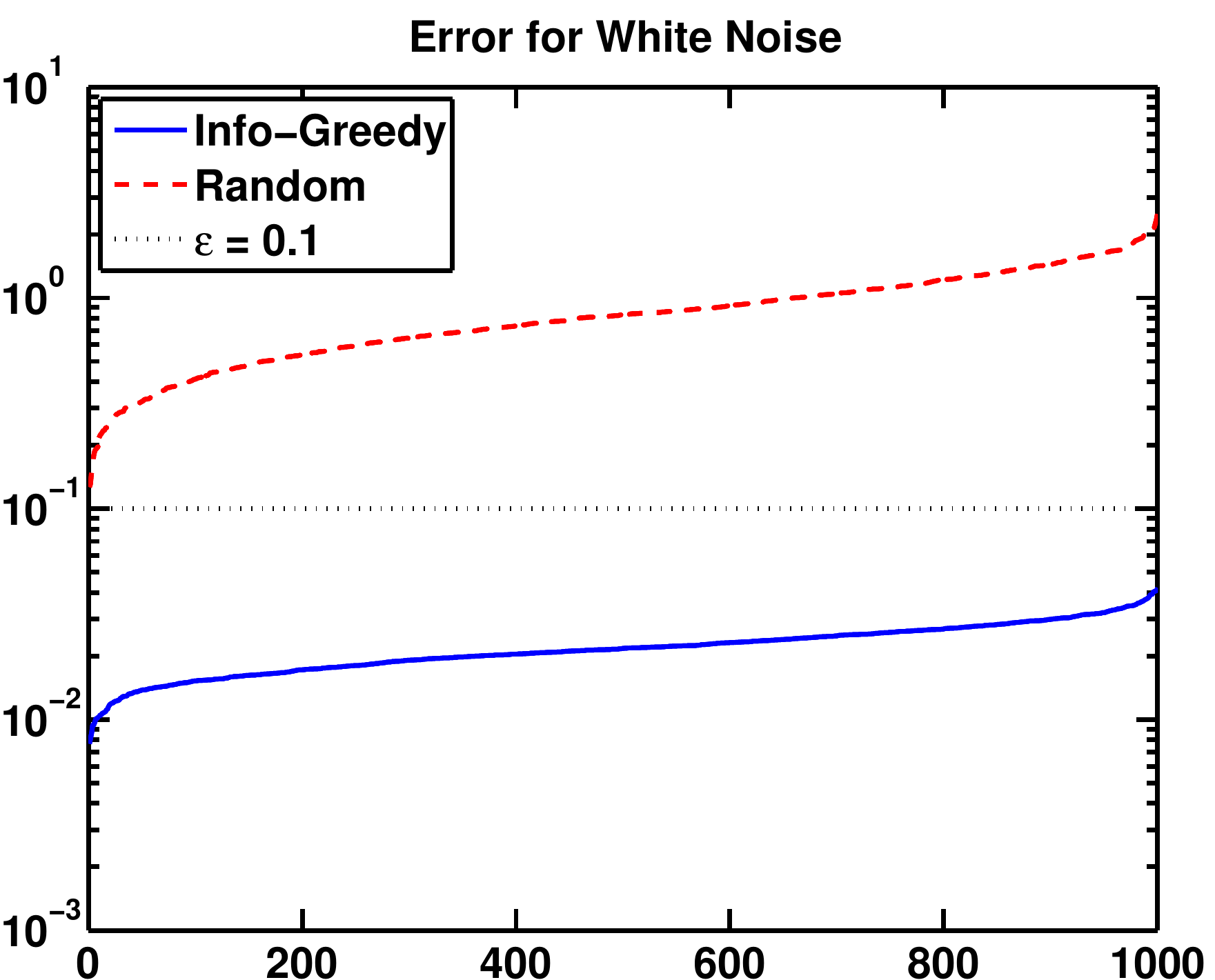} &
\includegraphics[width = 0.46\linewidth]{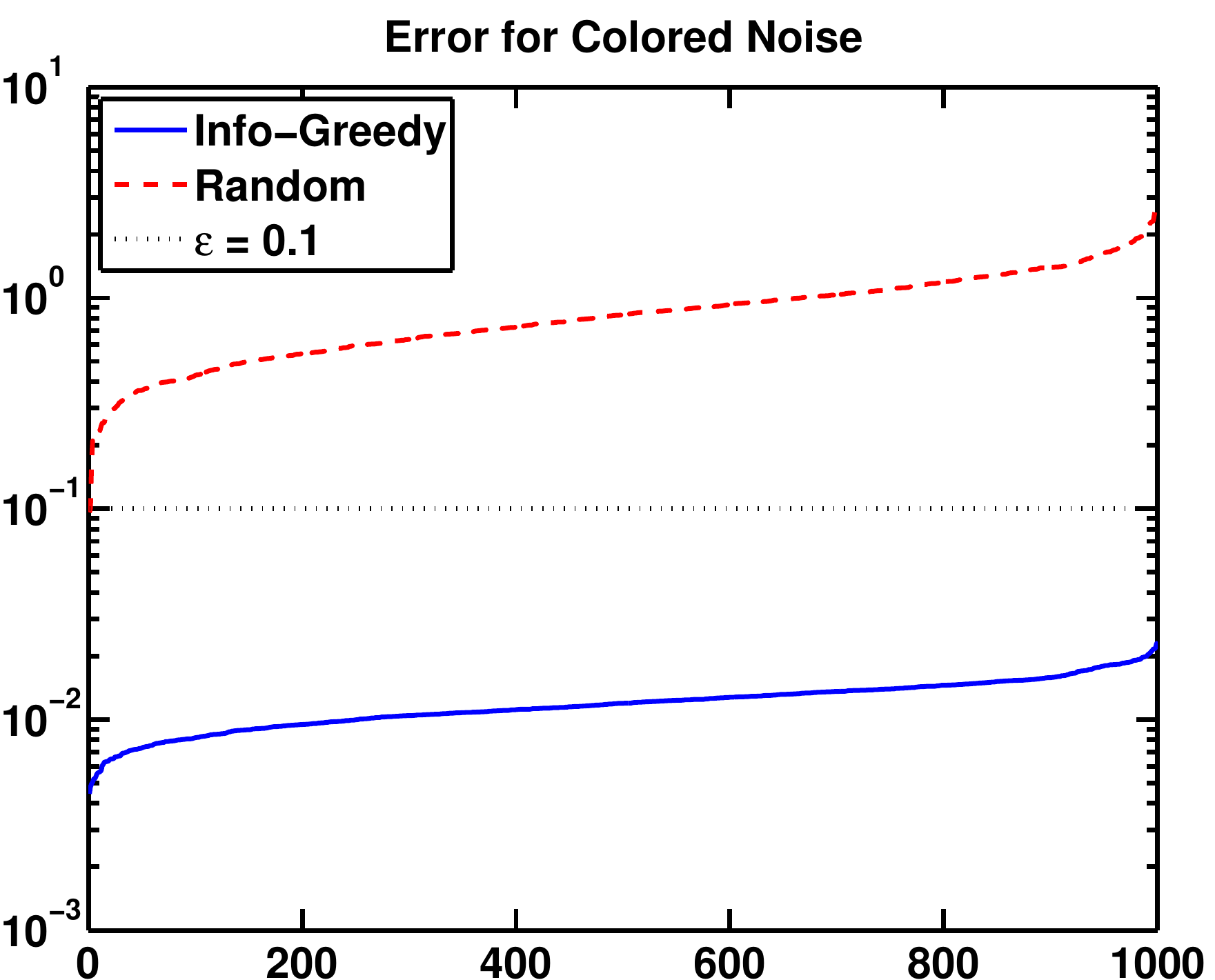} \\
(a) & (b) \\
\includegraphics[width = 0.45\linewidth]{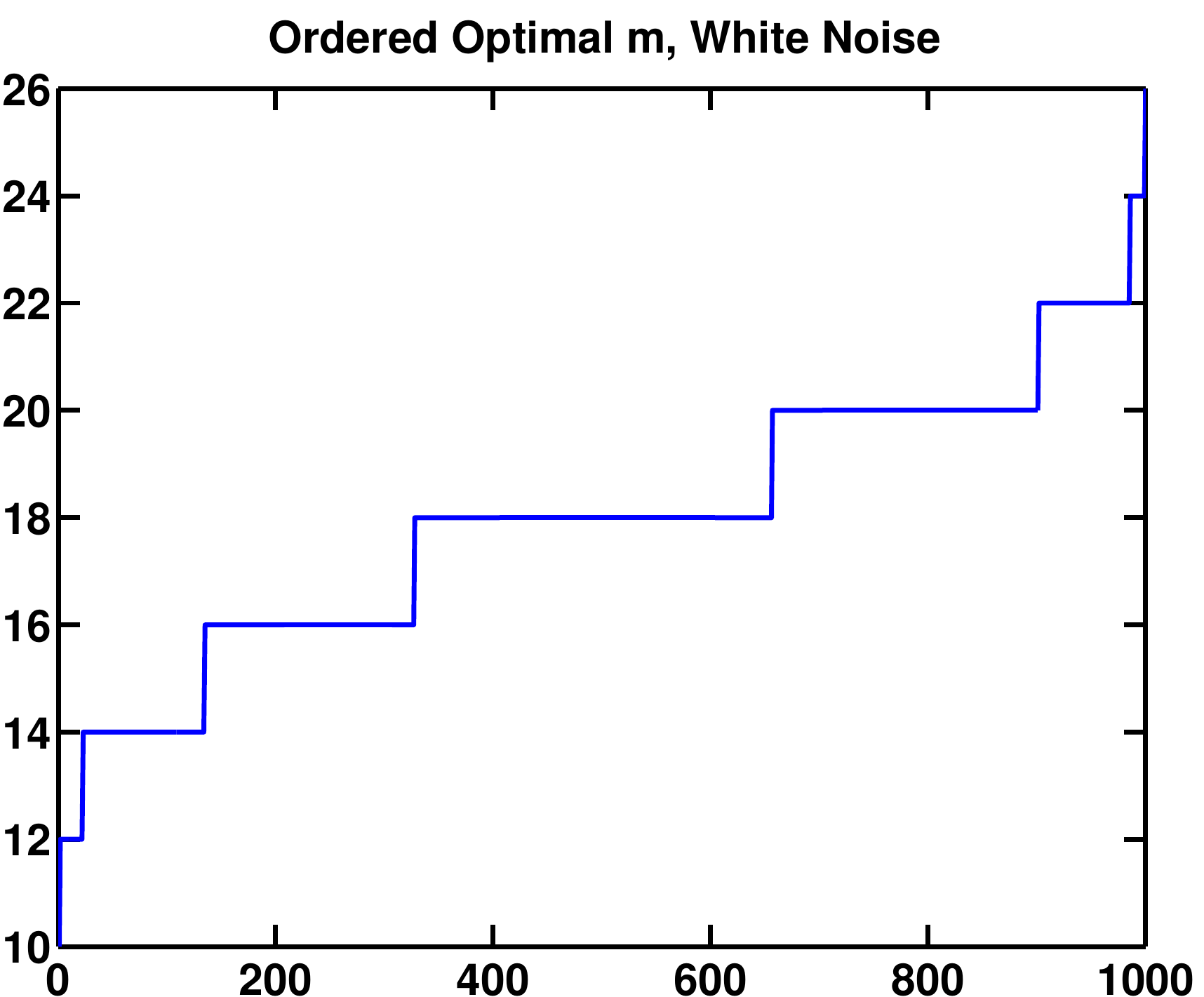} &
\includegraphics[width = 0.45\linewidth]{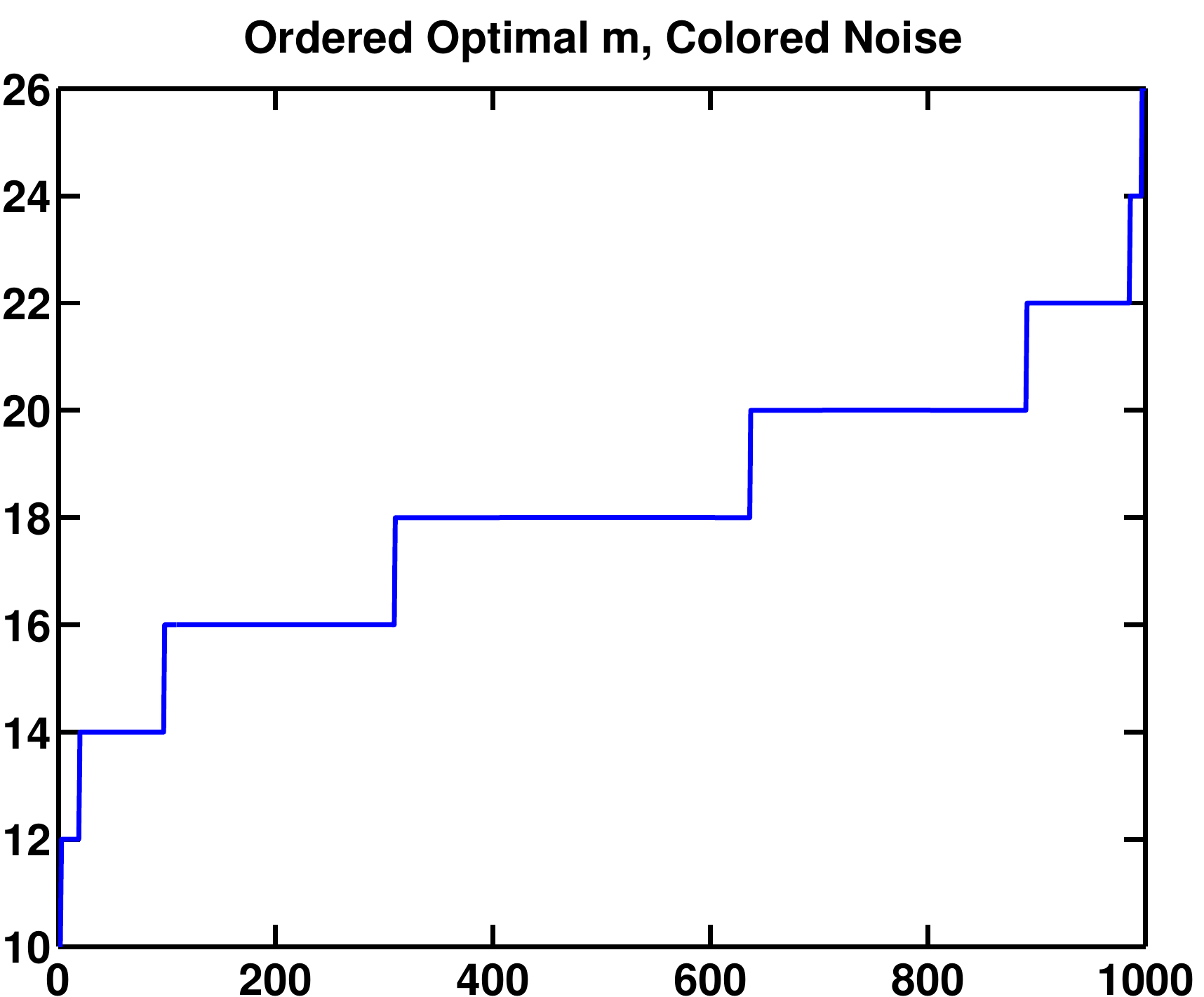} \\
(c) & (d)
\end{tabular}
\end{center}
\caption{Sensing a low-rank Gaussian signal of dimension $n=100$ and with about $70\%$ eigenvalues of $\Sigma$ zero: (a) and (c) compare recovery error $\|x - \widehat x\|_2$ for the Info-Greedy Sensing and random sensing $A$, in the presence of white noise added after the measurement, and colored noise added prior to the measurement (``noise folding''), respectively; (c) and (d) show ordered number of measurements for Info-Greedy Sensing in the two cases. \yx{Info-Greedy Sensing and batch method perform identical in this case.}}
\label{Fig:comp_white_color_Gaussian}
\end{figure}

\begin{figure}
\begin{center}
\includegraphics[width = 0.45\linewidth]{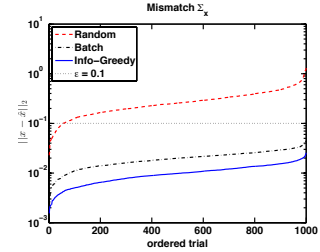}
\end{center}
\caption{\yx{Sensing a low-rank Gaussian signal of dimension $n= 500$ and about $5\%$ of the eigenvalues are non-zero, when there is mismatch between the assumed covariance matrix and true covariance matrix: ${\Sigma}_{\rm , assumed} = {\Sigma}_{\rm , true} + e \transpose{e}$, where $e\sim \mathcal{N}(0, I)$, and using 20 measurements. The batch method measures using the largest eigenvectors of ${\Sigma}_{\rm , assumed}$, and the Info-Greedy Sensing updates ${\Sigma}_{\rm , assumed}$ in the algorithm. Info-Greedy Sensing is more robust to mismatch than the batch method.}}
\label{Fig:mismatch}
\end{figure}

\begin{figure}
\begin{center}
\begin{tabular}{cc}
\includegraphics[width = 0.45\linewidth]{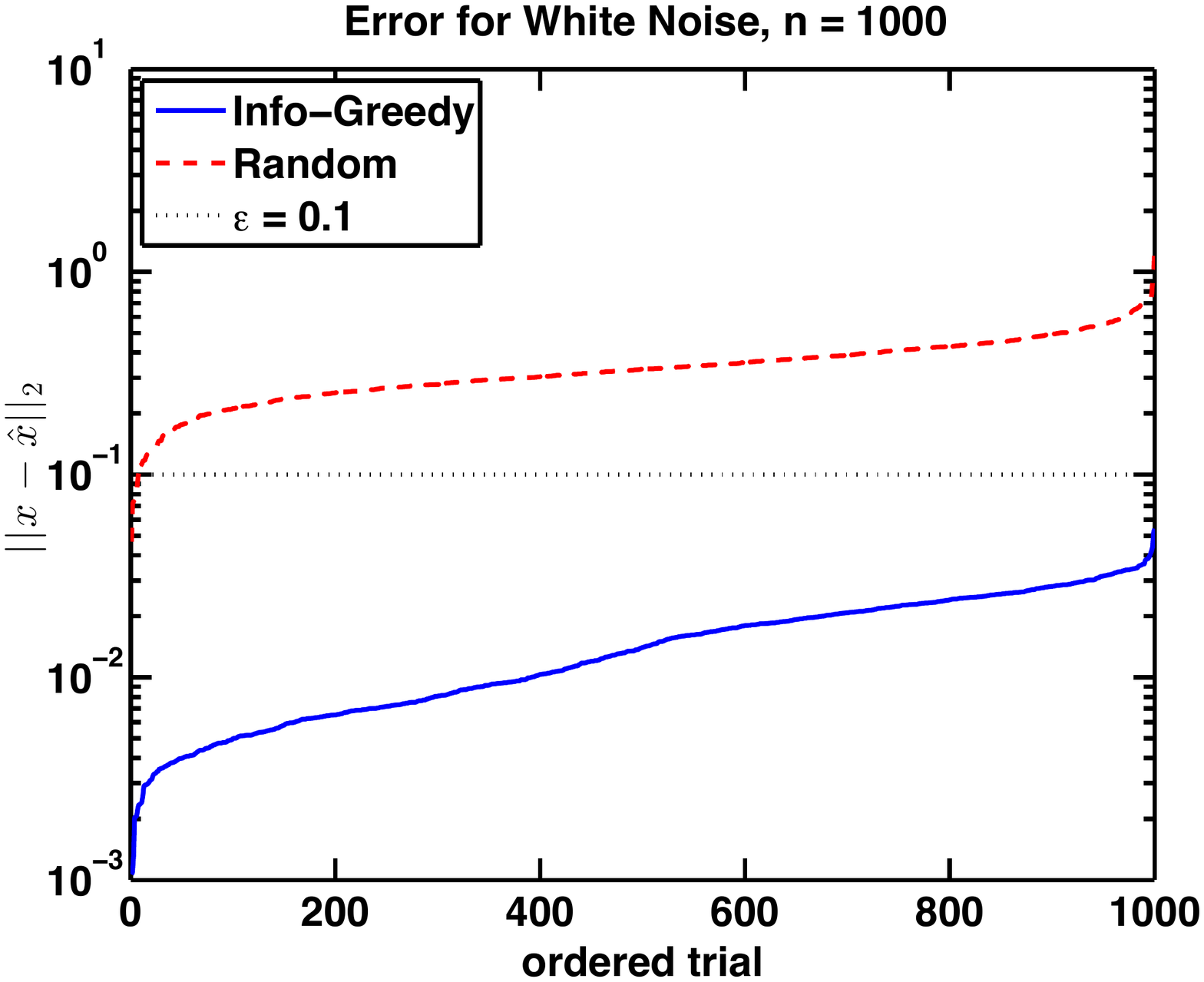} &
\includegraphics[width = 0.45\linewidth]{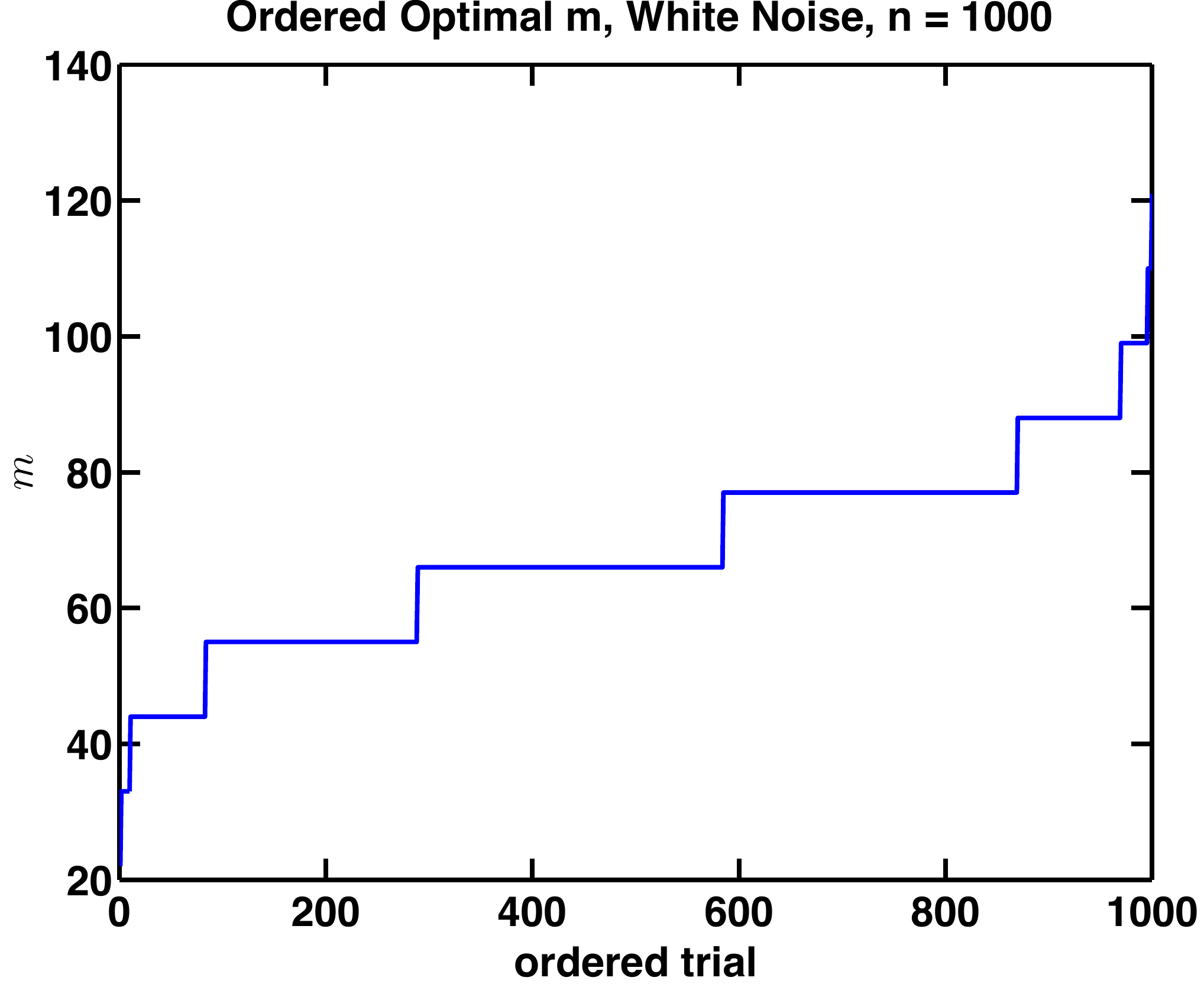}
\end{tabular}

\end{center}
\caption{Sense a low-rank Gaussian signal of dimension $n =1000$ and about $5\%$ eigenvalues of $\Sigma$ are non-zero. Info-Greedy Sensing has two orders of magnitude improvement over the random projection. 
}
\label{Fig:Gaussian_n1000}
\end{figure}

\begin{figure}
\begin{center}
\includegraphics[width = 0.45\linewidth]{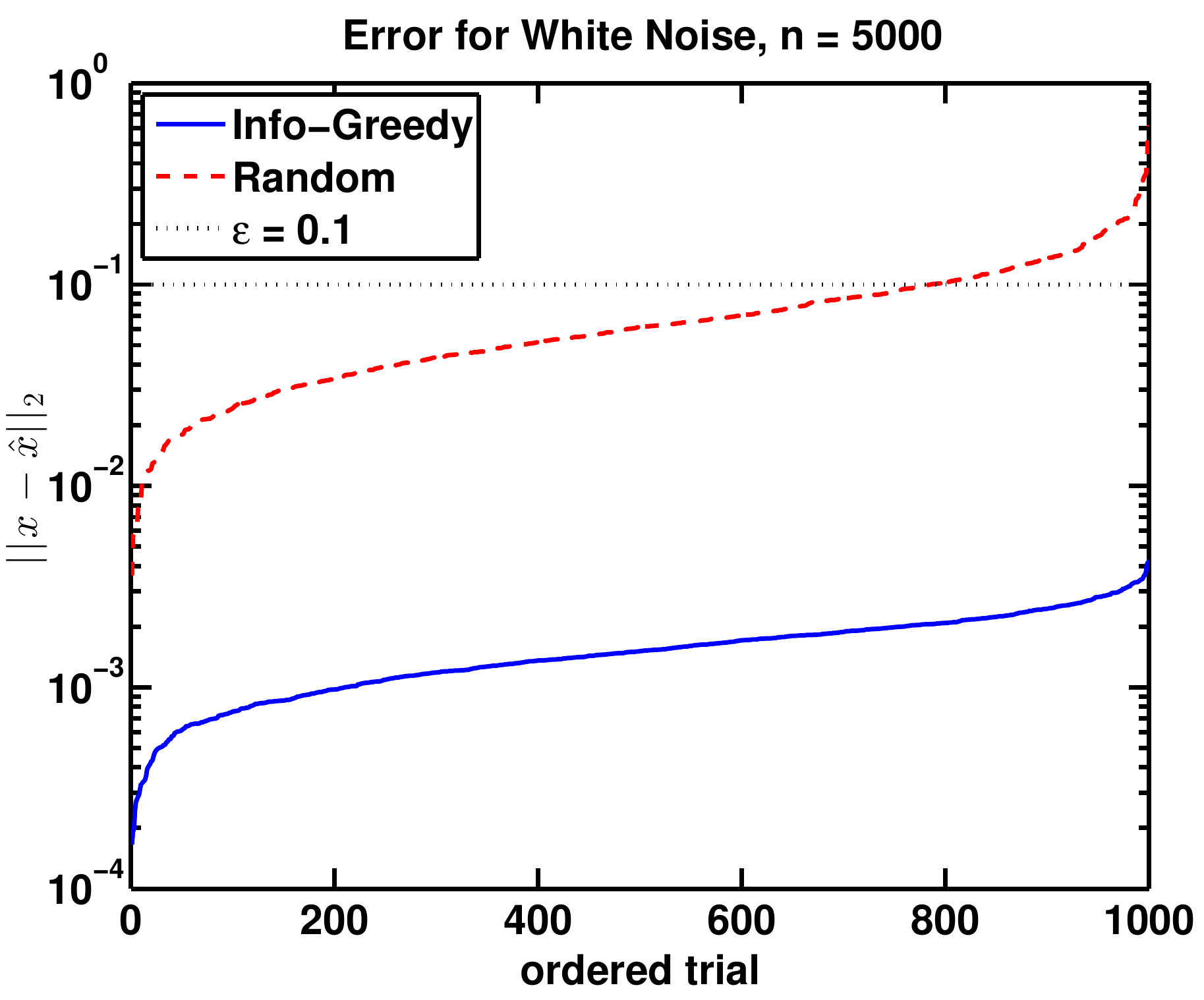} 

\end{center}
\caption{Sense a Gaussian signal of dimension $n = 5000$. The covariance matrix is low-rank and sparse: only $0.0003\%$ of entires $\Sigma$ are non-zero and the rank is 3. Info-Greedy Sensing has two orders of magnitude improvement over the random projection. The number of measurements is 3 as calculated through  (\ref{eq:Gaussian-number-fix-noise}).} 
\label{Fig:GaussianLarge}
\end{figure}

\subsubsection{Low-rank GMM model}

In this example we consider a GMM model with $C = 3$ components, and each Gaussian component is generated as a single Gaussian component described in the previous example Section \ref{sec:eg_Gaussian} ($n = 100$ and $\sigma = 0.01$). The true prior distribution is $\pi = (0.3, 0.2, 0.5)$ for the three components (hence each time the signal $x$ is draw from one component with these probabilities), and the assumed prior distribution for the algorithms is uniform: each component has probability $1/3$. The parameters for the gradient descent approach are: step size $\mu = 0.2$ and the error tolerance to stop the iteration $\eta = 0.01$. Fig. \ref{fig:MI} demonstrates the estimated cumulative mutual information and mutual information in a single step, averaged over 100 Monte Carlo trials, and the gradient descent based approach has higher information gain than that of the greedy heuristic, as expected. Fig. \ref{fig:errGMM} shows the ordered errors for the batch method based on mutual information gradient \cite{CarsonChenRodrigues2012}, the greedy heuristic versus gradient descent approach, when $m = 11$ and $m = 20$, respectively. Note that Info-Greedy Sensing approaches (greedy heuristic and gradient descent) outperform the batch method due to adaptivity, and that the simpler greedy heuristic  performs fairly well compared with the gradient descent approach. \yao{For GMM signals, the complexity of the batch method is $\mathcal{O}(C n^3)$ (due to eigendecomposition of $C$ components), versus the complexity of Info-Greedy Sensing algorithm is on the order of $\mathcal{O}(C t m n^2)$ where $t$ is the number of iterations needed to compute the eigenvector associated with the largest eigenvector (e.g., using the power method),  and $m$ is the number of measures which is typically on the order of $k$. }

\begin{figure}[h!]
\begin{center}
\begin{tabular}{cc}
\includegraphics[width = 0.45\linewidth]{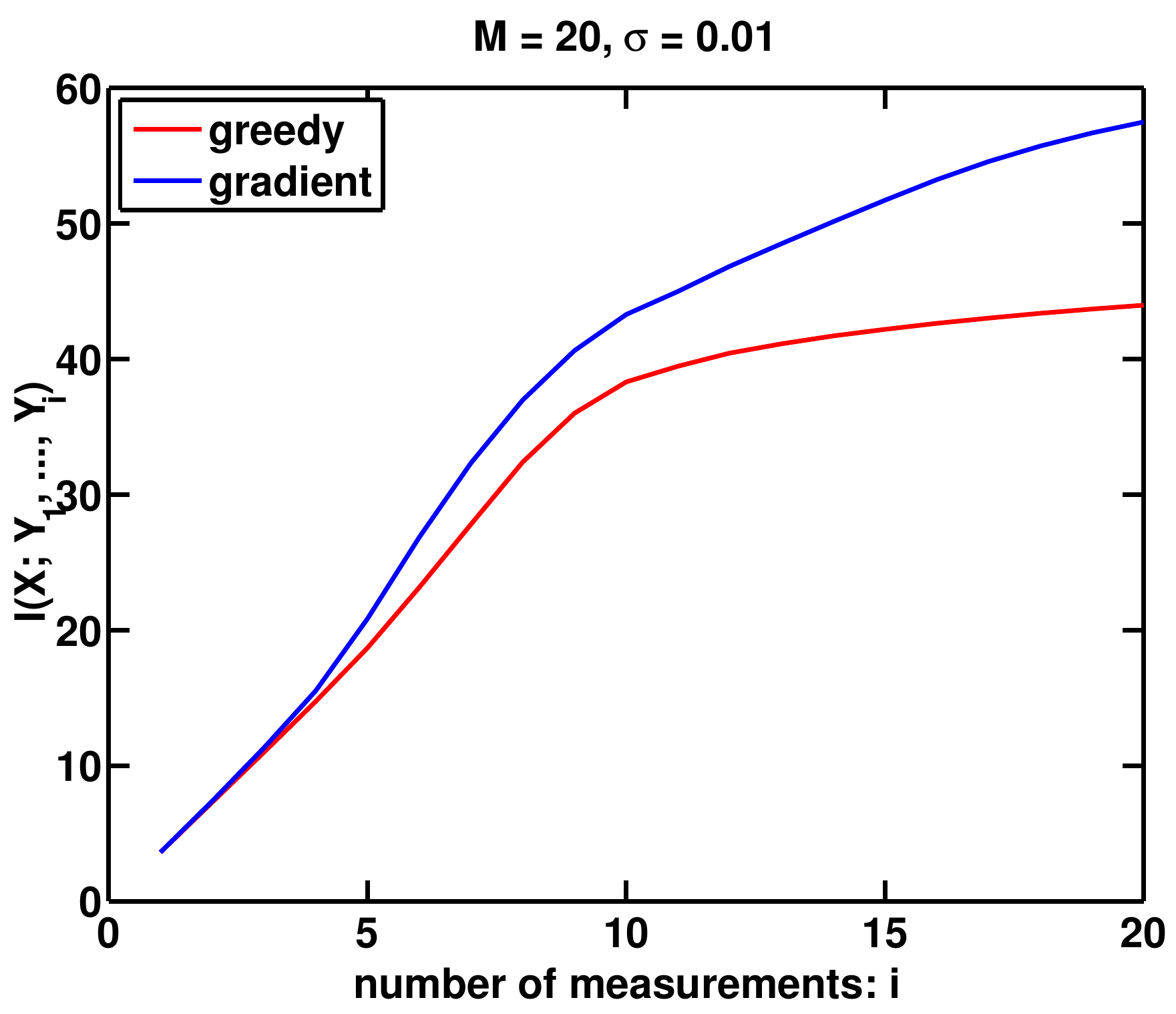} &
\includegraphics[width = 0.45\linewidth]{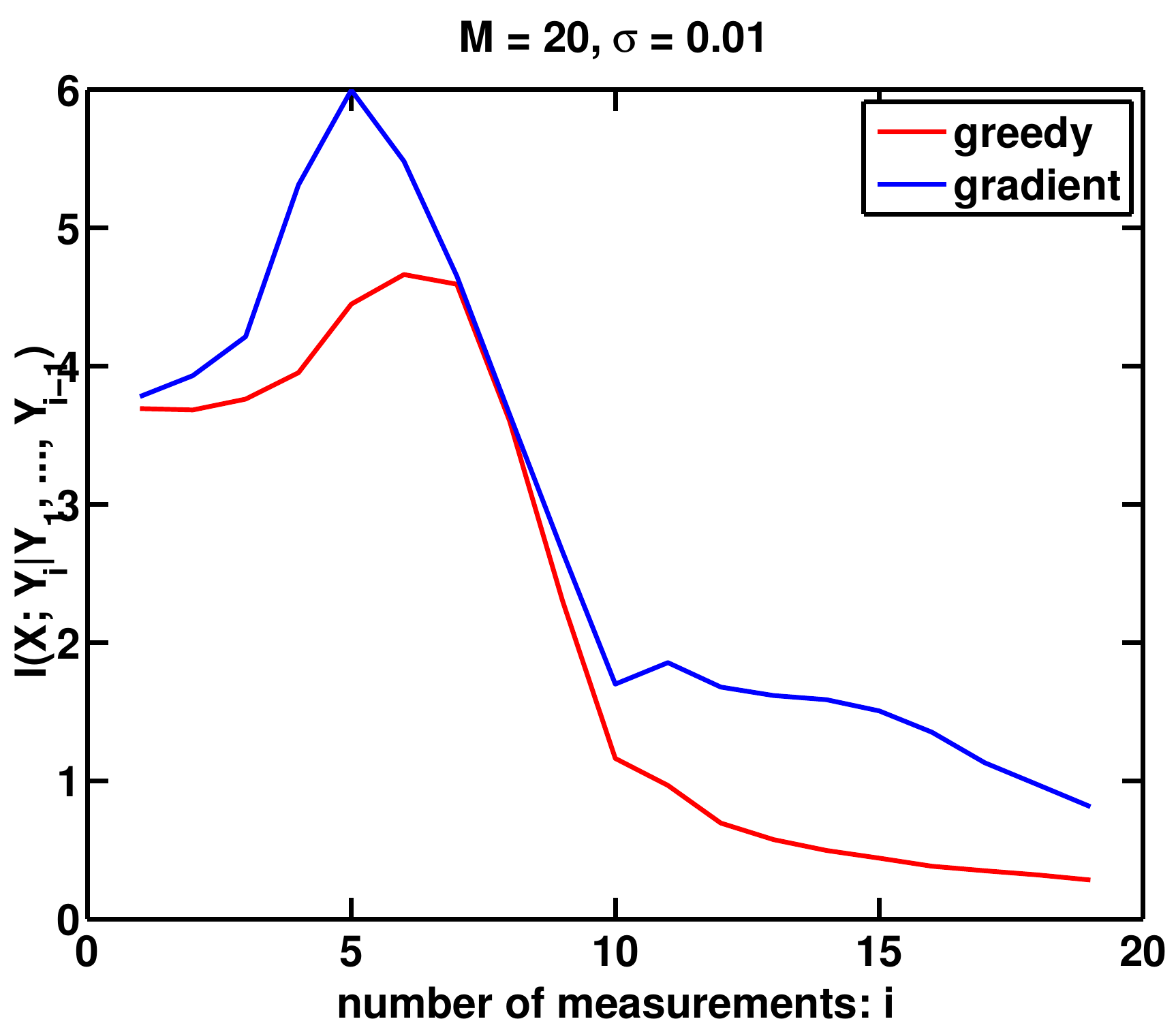} \\
(a) & (b)
\end{tabular}
\end{center}
\caption{Sensing a GMM signal: comparison of greedy heuristic and the gradient descent approach in terms of (a) mutual information $\mathbb{I}[x; y_1, \ldots, y_i]$ over number of measurements $i$, average over 100 Monte Carlo trials; (b) $\mathbb{I}[x; y_i|y_1, \ldots, y_i]$ over number of measurements $i$, averaged over 100 Monte Carlo trials. }
\label{fig:MI}
\end{figure}

\begin{figure}[h!]
\begin{center}
\begin{tabular}{c}
\includegraphics[width = 0.45\linewidth]{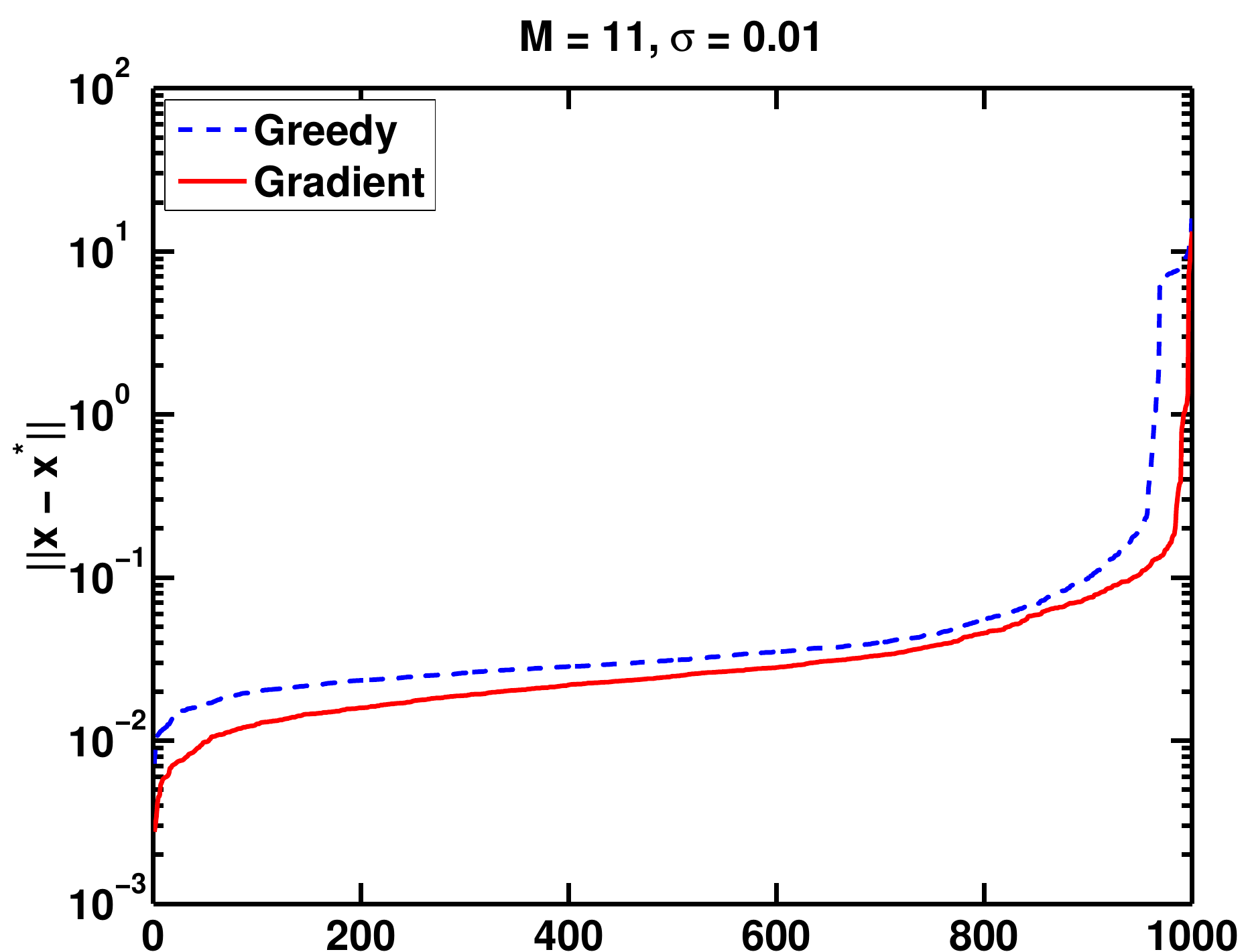}
\includegraphics[width = 0.45\linewidth]{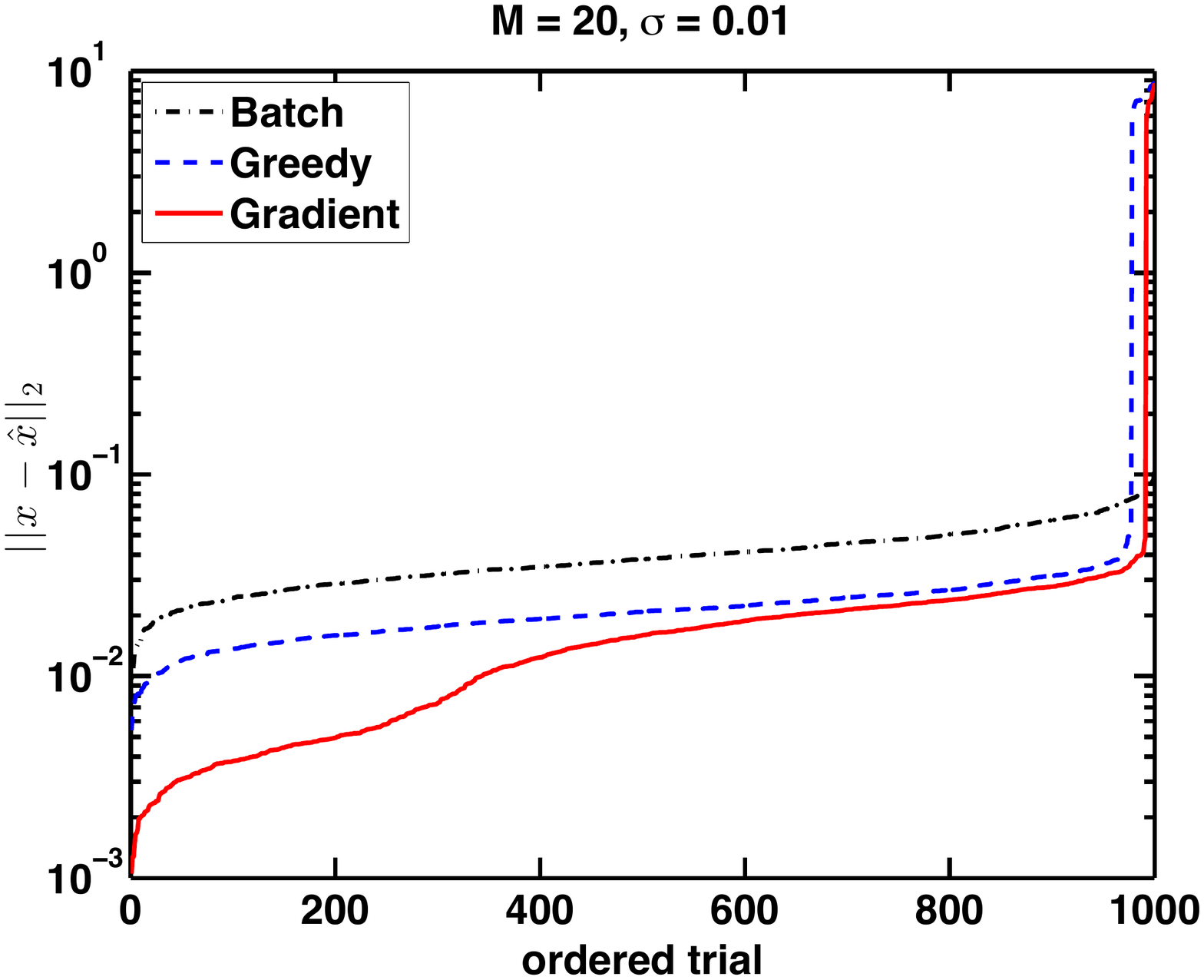}
\end{tabular}
\end{center}
\caption{Sensing a GMM signal: comparison of errors for \yx{the batch gradient descent method \cite{CarsonChenRodrigues2012}} and the Info-Greedy Sensing algorithms: the greedy heuristic and the gradient descent approach, when $m = 11$ and $m = 20$, respectively. }
\label{fig:errGMM}
\end{figure}

\subsubsection{Sparse Info-Greedy Sensing}

Consider designing a sparse Info-Greedy Sensing vector for a single Gaussian signal with $n = 10$, desired sparsity of measurement vector $k_0 = 5$, and the low-rank covariance matrix is generated as before by thresholding eigenvalues. 
Fig. \ref{fig:sparse}(a) shows the pattern of non-zero entries from measurement 1 to 5. Fig. \ref{fig:sparse}(b) compares the performance of randomly selecting 5 non-zero entries. The sparse Info-Greedy Sensing algorithm outperforms the random approach and does not degrade too much from the non-sparse Info-Greedy Sensing.  

\begin{figure}
\begin{center}
\begin{tabular}{cc}
\includegraphics[width = 0.45\linewidth]{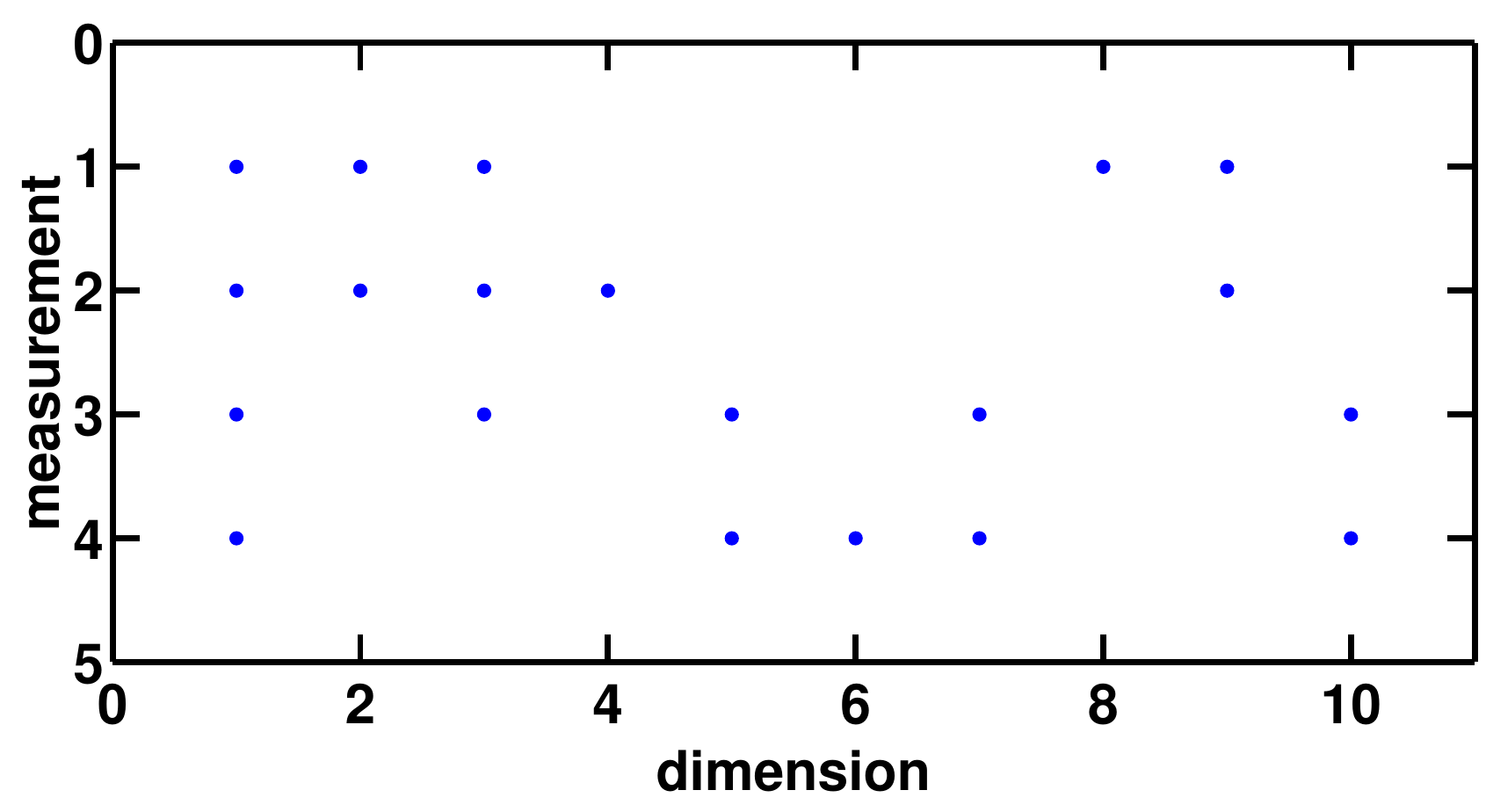} & \includegraphics[width = 0.45\linewidth]{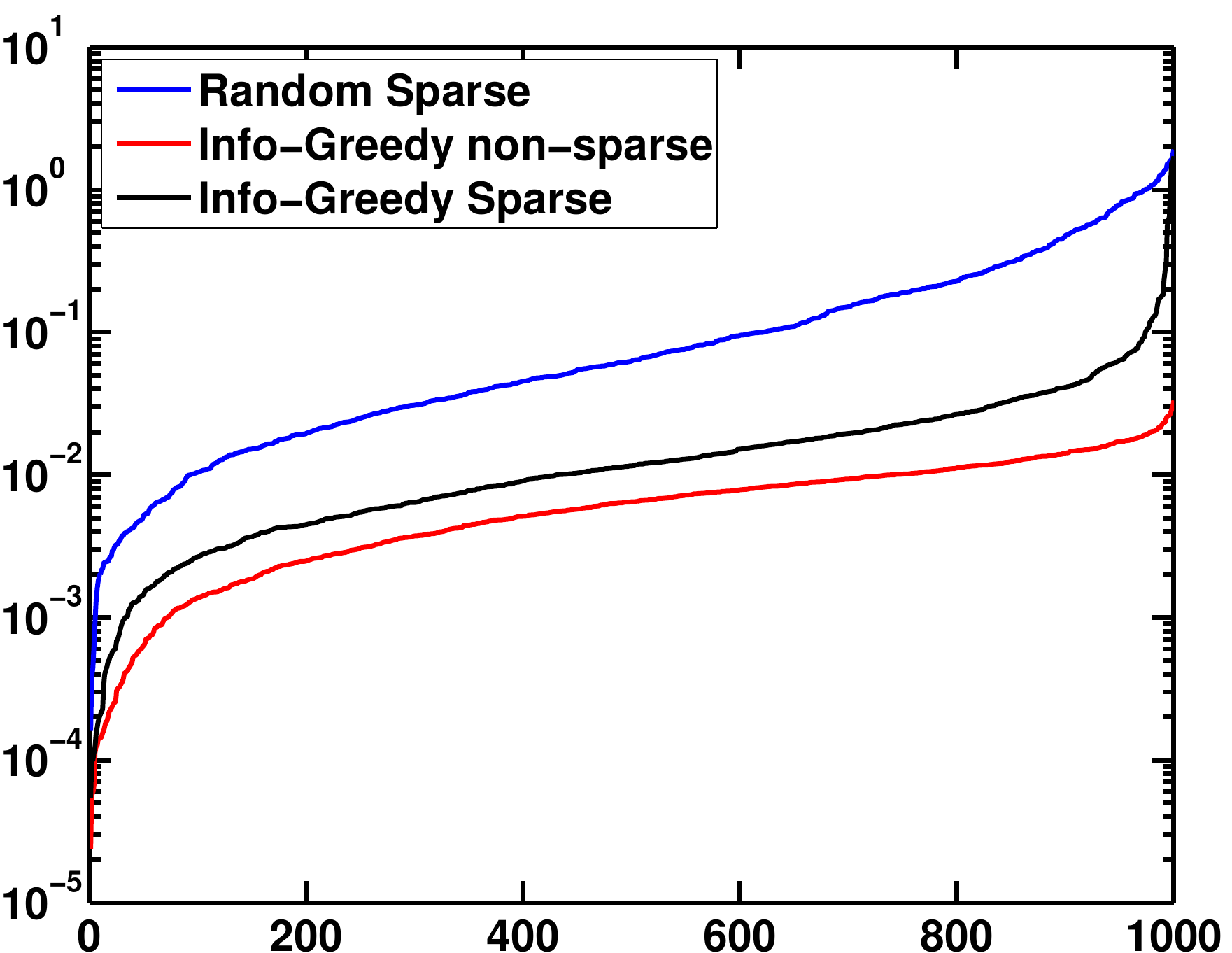} \vspace{0.1in}\\
(a) & (b)
\end{tabular}
\end{center}
\caption{Results of designing sparse sensing vectors: (a) support of the sparse measurements for $n = 10$, $k_0 = 5$, over 5 measurements; (b) comparison of errors for the random sparse measurement, sparse Info-Greedy  measurement, and non-sparse Info-Greedy measurement. }
\label{fig:sparse}
\end{figure}

\subsection{Real data}


\subsubsection{MNIST handwritten dataset}

We exam the performance of using GMM Info-Greedy Sensing on MNIST handwritten dataset\footnote{http://yann.lecun.com/exdb/mnist/}. In this example, since the true label of the training data is known, we can use training data to estimate the true prior distribution $\pi_c$, $\mu_c$ and $\Sigma_c$ (there are $C = 10$ classes of Gaussian components each corresponding to one digit) using 10,000 training pictures of handwritten digits picture of dimension 28 by 28. The images are vectorize and hence $n = 784$, and the  digit can be recognized using the its highest posterior $\pi_c$ after sequential measurements. Fig. \ref{Fig:MNIST} demonstrates an instance of recovered image (true label is 2) using $m  =40$ sequential measurements, for the greedy heuristic and the gradient descent approach, respectively. In this instance, the greedy heuristic classifies the image erroneously as 6, and the gradient descent approach correctly classifies the image as 2.
Table \ref{table_Pe} shows the probability of false classification for the testing data, where the random approach is where $a_i$ are normalized random Gaussian vectors. Again, the greedy heuristic  has good performance compared to the gradient descent method. 

\begin{figure}[h!]
\begin{center}
\includegraphics[width = .5\linewidth]{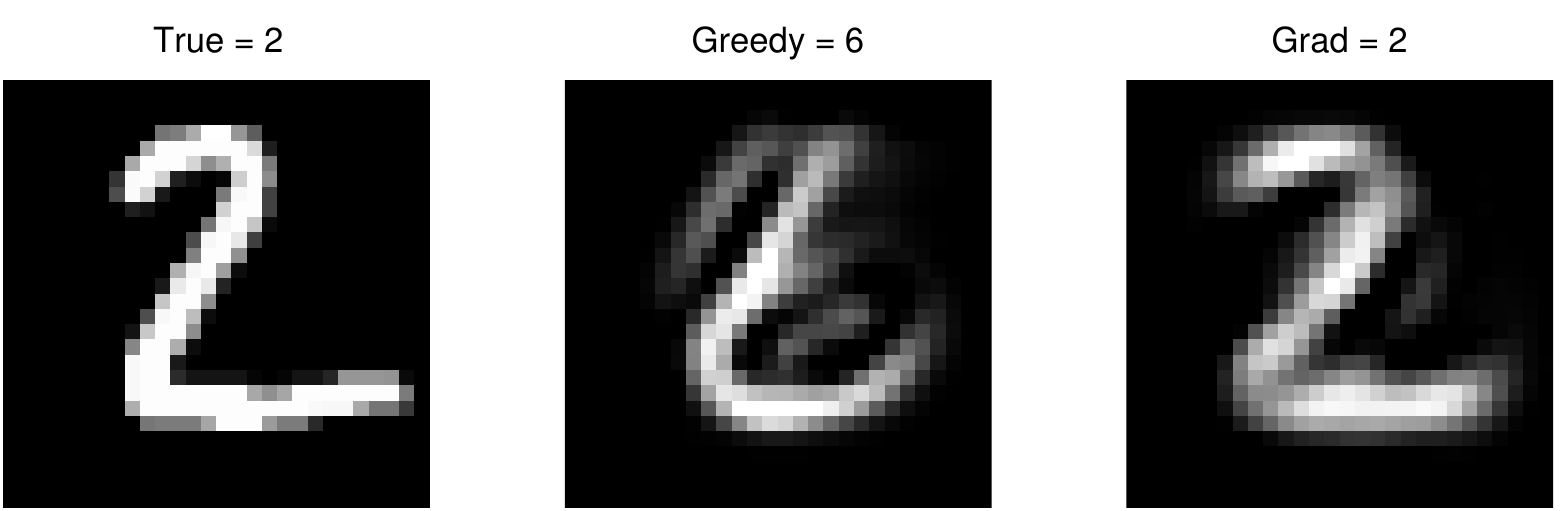}
\end{center}
\caption{Comparison of true and recovered handwritten digit 2 by the greedy heuristic and the gradient descent approach, respectively.}
\label{Fig:MNIST}
\end{figure}

\begin{table}[h!]
\begin{center}
\caption{Comparison of probability of false classification for MNIST handwritten digits dataset.} \label{table_Pe}
\begin{tabular}{c|ccc}
\hline
 Method & Random & Greedy & Gradient \\\hline
 prob. false classification &  0.192 & 0.152 &  0.144  \\\hline
\end{tabular}
\end{center}
\end{table}

\subsubsection{Recovery of power consumption vector}

We consider recovery of a power consumption vector for 58 counties in California\footnote{http://www.ecdms.energy.ca.gov/elecbycounty.aspx }. Data for power consumption in these counties from year 2006 to year 2012 are available. We first fit a single Gaussian model using data from year 2006 to 2011 (Fig. \ref{fig:power}(a), the probability plot demonstrates that Gaussian is a reasonably good fit to the data), and then test the performance of the Info-Greedy Sensing in recovering the data vector of year 2012. Fig. \ref{fig:power}(b) shows that even by using a coarse estimate of the covariance matrix from limited data (5 samples), Info-Greedy Sensing can have better performance than the random algorithm. This example has some practical implications: the compressed measurements here correspond to collecting the total power consumption over a region of the power network. This collection process can be achieved automatically by new technologies such as the wireless sensor network platform using embedded RFID in \cite{WirelessHouseElectricity2014} and, hence, our Info-Greedy Sensing may be an efficient solution to monitoring of power consumption of each node in a large power network. 

\begin{figure}[h!]
\begin{center}
\begin{tabular}{cc}
\includegraphics[width = .47\linewidth]{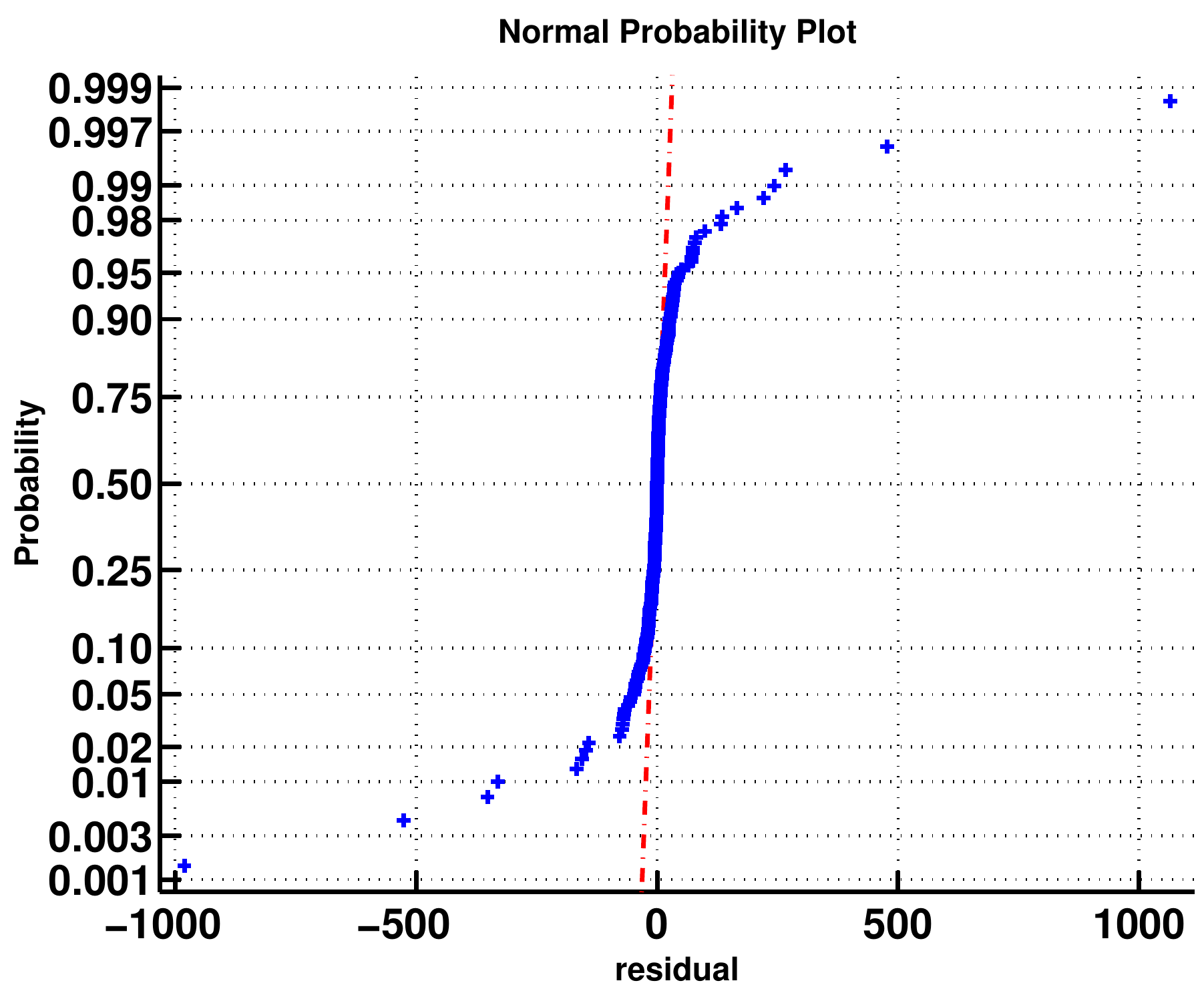} &
\includegraphics[width = .47\linewidth]{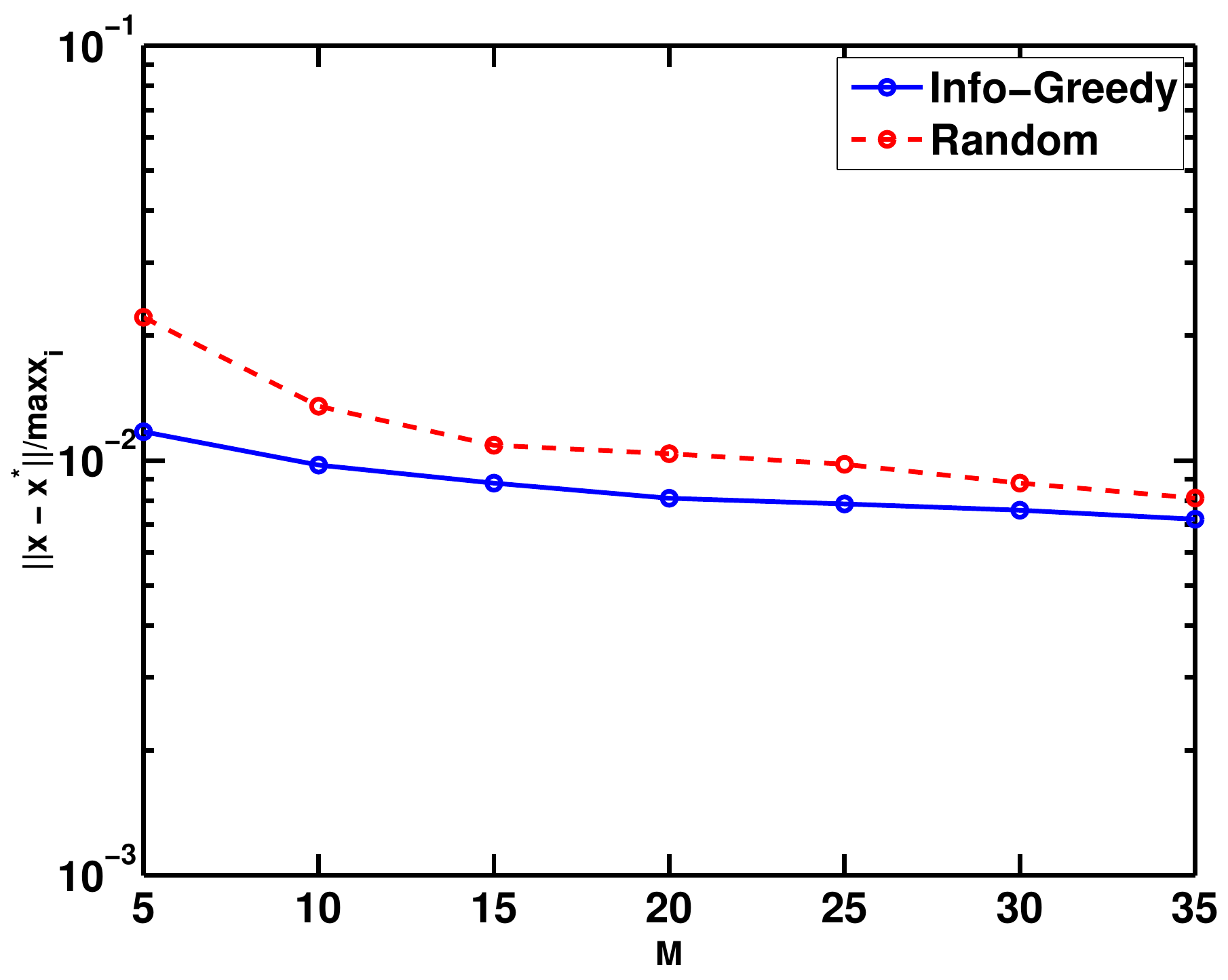}
\\
(a) &  (b) 
\end{tabular}
\end{center}
\caption{Recovery of power consumption data of 58 counties in California: (a) normal probability of residuals formed by training data after subtracting the mean estimated from year 2006 to year 2011; (b) relative error $\|x - \hat{x}\|_2/\|x\|_\infty$ for estimating power consumption vector in year 2012 versus the number of measurements.
}
\label{fig:power}
\end{figure}

\section{Conclusion}\label{sec:con}

We have presented a general framework for sequential adaptive compressed sensing, Info-Greedy Sensing, which is based on maximizing mutual information between the measurement and the signal model conditioned on previous measurements. Our results demonstrate that adaptivity helps  when prior distributional information of the signal is available and Info-Greedy is an efficient tool to explore such prior information, such as in the case of the GMM signals. Adaptivity also brings robustness when there is mismatch between the assumed and true distribution, and we have demonstrated such benefits for Gaussian signals.  Moreover, Info-Greedy Sensing shows significant improvement over random
projection for signals with sparse and low-rank covariance matrices, which demonstrate the potential value of Info-Greedy Sensing for big data.

\bibliographystyle{ieeetr}
\bibliography{bib}

\appendices

\section{General performance lower bounds} \label{sec:bound}

In the following we establish a general lower bound for the number of sequential measurements needed to obtain certain small recovery error \(\|x - \widehat{x}\|_2\), similar to the approach in \cite{bgp2013}. We consider the following model: sequentially perform measurements and performance is measured by the number \(M\) of measurements required to obtain a reconstruction of the signal with a prescribed accuracy. 
Assume the sequential measurements \(a_i\) are linear and the measurement returns
\(\transpose{a}_i x\). 
Formally, let \(\mathcal F\) be a \yx{finite} family of signals of
interest, and \(F \in \mathcal F\) be a random variable with uniform
distribution on $\mathcal F$.
Denote by \(A=(a_1,a_2,\dotsc)\) the sequence of measurements, and
\(y=(y_1,y_2,\dotsc)\) the sequence of measurement values:
\(y_i = \transpose a_i x\). Let \(\Pi = (A,y)\) denote the transcript of the measurement operations and \(\Pi_i = (a_i,y_i)\) a single measurement/value pair. Note that \(\Pi\) is a random variable of the picked signal \(F\). Assume that the accuracy \(\varepsilon\) is high enough to ensure a one-to-one
correspondence between signal \(F\) and the \(\varepsilon\)-ball it is
contained in. Thus we can return the center of such an \(\varepsilon\)-ball  as the reconstruction \(\widehat x\) of \(x\). 
In this regime, an \(\varepsilon\)-recovery of a signal $x$ is
(information-theoretically) equivalent to
learning the \(\varepsilon\)-ball that $x$ is contained in, 
and we can invoke the \emph{reconstruction principle}
\begin{equation}
  \label{eq:conservation}
  \mutualInfo{F}{\Pi} = \entropy{F} = \log \card{\mathcal F},
\end{equation}
i.e., the transcript has to contain the same information as \(F\) and in fact uniquely identify it.
With this model it was shown in \cite{bgp2013} that the total amount of information acquired, \(\entropy{F}\), is equal to the sum of the conditional information per iteration: 
\begin{thm}[\cite{bgp2013}] 
\label{thm:estimateRunningTime}
  \begin{equation}
    \label{eq:1}
\mutualInfo{F}{\Pi} = \sum_{i=1}^{\infty}
  \underbrace{\entropy[a_i,\Pi^{i-1}, M \geq i]{y_i}}%
  _{\text{\rm information gain by measurement $i$}}
  \probability{M\geq i},
  \end{equation}
where \(\Pi^{i-1}\) is a shorthand for \(\Pi^{i-1} \triangleq
(\Pi_1,\dotsc,\Pi_{i-1})\) and \(M\) is the random variable of required measurements. 
\end{thm}
We will use Theorem \ref{thm:estimateRunningTime} 
to establish Lemma \ref{lem:bisection-k-ub} that the bisection
  algorithm is Info-Greedy for $k$-sparse signals.
\emph{A priori}, Theorem \ref{thm:estimateRunningTime} does not give a bound on the expected number of required measurements, and it only characterizes how much information the sensing algorithm learns from each measurement. However, if we can upper bound the information acquired in each measurement by some constant, this  leads to a lower bound on the \emph{expected} number of measurements, as well as a \emph{high-probability} lower bound:
    \begin{cor}[Lower bound on number of measurements] \label{thm:lower_bound}
Suppose that for some constant \(C > 0\),
\[\entropy[a_i,\Pi^{i-1}, m \geq i]{y_i} \leq C\]
for every round \(i\) where $M$ is as above.
Then \(\expectation{M} \geq \frac{\log \card{\mathcal F}}{C}\).
Moreover, for all \(t\) we have \(\probability{M
  < t} \leq (Ct)/\entropy{F}\) and 
\(\probability{T
  = \mathcal{O}(\entropy{F})} = 1 - o(1)\).
    \end{cor}

\yx{The information theoretic approach also lends itself to
lower bounds on the number of measurements for Gaussian signals,
as e.g., in \cite[Corollary~4]{cont_Fano}.}

\section{Derivation of Gaussian signal measured with colored noise}\label{app:Gaussian_color}

First consider the case when colored noise is added after the measurement: $y = Ax + w$, $w\sim \mathcal{N}(0, \Sigma_w)$. In the following, we assume the noise covariance matrix \(\Sigma_{w}\) is full rank. Note that we can write $w_1 = \transpose e_1 w$. Let the eigendecomposition of the noise covariance matrix be $\Sigma_w = U_w \Lambda_w \transpose U_w$, and define a constant vector $b \triangleq \Lambda_w^{1/2} \transpose{U}_w e_1$. So the variance of $w_1$ is given by $\transpose e_1 \Sigma_w e_1 = \transpose b b$. Re-parameterize $a_1$ by introducing a unitary matrix $R$: 
$a_1 =  (\sqrt{\beta_1}/\|b\|_2) R b$.
%
Also let the eigendecomposition of $\Sigma$ be $\Sigma = U_x \Lambda_x \transpose U_x$.
Then the mutual information of \(x\) and \(y_{1}\) can be written as
\begin{equation}
\begin{split}
&  \mutualInfo{x}{y_{1}}
  = \frac{1}{2}\ln
  \left(
    \frac{\transpose{a_{1}} \Sigma a_{1}}
    {\transpose e_1 \Sigma_w e_1}
    + 1
  \right)    =  \frac{1}{2}\ln
  \left(   \frac{\beta_1}{\|b\|_2^2} \cdot
    \frac{\transpose b \transpose R \Sigma Rb}
    {\transpose b b}
    + 1
  \right)\\
  & = \frac{1}{2}\ln
  \left(\frac{\beta_1}{\|b\|_2^2} \cdot
    \frac{\transpose e_1 U_w \Lambda_w^{1/2} \transpose R U_x \Lambda_x \transpose U_x  R  \Lambda_w^{1/2} \transpose{U}_w e_1}
    {\transpose b b}
    + 1
  \right) \\
  & \leq \frac{1}{2} \ln
  \left(
  \frac{\beta_1}{\|\Lambda_w^{1/2} \transpose{U}_w e_1 \|_2^4} 
  \cdot \|\Sigma\|\|\Sigma\|
  \right),
  \end{split}
  \label{eq:gGaussian-noise-mutual}
\end{equation}
and the maximum is achieved when $R = U_x$. Hence, the Info-Greedy Sensing vector is 
\begin{equation}
a_1 = \frac{\sqrt{\beta}_1}{\|\Lambda_w^{1/2} \transpose U_w e_1\|_2} U_x \Lambda_w^{1/2} \transpose U_w e_1. \label{matching_a}
\end{equation}
Note that the solution (\ref{matching_a}) for $a_1$ has the interpretation of ``mode matching'', i.e.,  
aligning of eigenspaces of the signal and the noise similar to that in  \cite{CarsonChenRodrigues2012} for the non-adaptive setting. 

For the ``noise folding'' model with colored noise, $y = A (x + w)$, $w \sim \mathcal{N}(0, \Sigma_w)$, since the power of $a_i$ does not affect SNR, we assume $\|a_1\|_2 = 1$.  Let \(d \triangleq \Lambda_w^{1/2} \transpose U_w {a_{1}}\) and, hence, \(a_1 = U_w \Lambda_w^{-1/2} d\).  In this case
\begin{equation}
\begin{split}
  & \mutualInfo{x}{y_{1}}
   = \frac{1}{2}\ln
  \left(
    \frac{\transpose{a_{1}} \Sigma a_{1}}
    {\transpose{a}_1 \Sigma_w a_1}
    + 1
  \right) \\
  & = \frac{1}{2}\ln
  \left(
    \frac{ \transpose d  \Lambda_w^{-1/2} \transpose U_w
     \Sigma  U_w \Lambda_w^{-1/2}d}
    {\transpose d d}
    + 1
  \right)  \leq 
  \frac{1}{2}\ln
  \left(
    \|\Sigma'\|
    + 1
  \right),
  \end{split}
\end{equation}
where $\Sigma' = \Lambda_w^{-1/2} \transpose U_w
     \Sigma  U_w \Lambda_w^{-1/2}$, 
and the maximum is achieved when $d$ is the eigenvector for the largest eigenvector of $\Sigma'$. Equivalently \(a_{1}\)
is an eigenvector for the largest eigenvalue of
\(U_w\Lambda_w^{-1/2} \Sigma'_x \Lambda_w^{-1/2} U_w\) or the largest eigenvector for \(\Sigma_{w}^{-1} \Sigma\). Note that in contrast to (\ref{matching_a}), in this case the ``mode matching'' is not possible because the noise covariance depends on the measurement vector $a_i$ as well.

\section{Derivation for GMM signals}\label{app:GMM}

 Let \(\widetilde{p}\) and $\widetilde{\mathbb{E}}$ denote the probability density function and expectation given $\{y_{j} : j < i\}$. Using \cite{PalomarVerdu2006,MChenThesis,CarsonChenRodrigues2012},
the gradient of mutual information with respect to $a_i$ is given by
\begin{equation}
\frac{\mutualInfo[y_j, j < i]{x}{y_i}}{\partial a_i} 
= \frac{\transpose{E_i(a_i; y_j, j < i)} a_i }{\sigma^2},
\label{gradient}
\end{equation}
where ${E}_i(a_i; y_j, j < i) \in\mathbb{R}^{n\times n}$ is the MMSE matrix conditioned on measurements prior to $i$, which can be written as
\begin{equation}
\begin{split}
  &E_i 
  = \int \widetilde{p}(y) \cdot
  \\
  &
  \underbrace{\int \widetilde{p}(x \mid y_{i} = y)
    (x - \widetilde{\mathbb{E}}[x\mid y_i = y])
    \transpose{(x - \widetilde{\mathbb{E}}[x\mid y_i = y])} dx}_{g(y)} dy.
    \end{split}
\label{Ei_def}
\end{equation}
For GMM, a closed form formula for the integrand $g(y)$ can be derived. %
Note that the conditional distribution of
\(x\) given the $\{y_{j} : j < i\}$ and \(y_{i} = y\) turns out to be 
a GMM with updated parameters: mean \(\widetilde{\mu}_{c}\),
variance \(\widetilde{\Sigma}_c\), and weight \(\widetilde{\pi}_{c}\):
\begin{align}
  \widetilde{\mu}_{c}(y)
  &
  =
  \mu_{c} + \Sigma_{c} \transpose{D_{i}}
    (\widetilde{y}_{i} - D_{i} \mu_{c})/\sigma^{2},
  \\
  \widetilde{\Sigma}_c
  &
  =
  \Sigma_{c}
  - \Sigma_{c} \transpose{D_{i}} D_{i} \Sigma_{c} / \sigma^{2},
  \\
  \widetilde{\pi}_c
  &
  \propto
 \pi_c \Phi(\widetilde{y}_{i}; D_i \mu_c,
      D_i \Sigma_c \transpose{D_i} + \sigma^2),
\end{align}
where 
 \(\transpose{D_i}
  =
   [ a_{1},  \cdots, a_{i-1}, a_{i}]
\) 
and 
\( 
  \widetilde{y}_i
  =
  \transpose{
   [ y_1, \cdots, y_{i-1}, y ]
  },
\)
and hence  
$
\widetilde{\mathbb{E}}[x \mid y_{i} = y, c] =
   \widetilde{\mu}_{c}(y)$, 
$\widetilde{\mathbb{E}}[x\mid y_{i} = y]  =
  \sum_{c=1}^{C} \widetilde{\pi}_{c} \widetilde{\mu}_{c}(y)
$.
Based on the above results
\begin{equation}
\begin{split}
&g(y)  
     \\
   \hspace{-0.05in} = & \sum_{c = 1}^C  \widetilde{\pi}_c
    \{
    \widetilde{\Sigma}_c +(\widetilde{\mu}_{c}(y) - \sum_{c=1}^{C} \widetilde{\pi}_{c} \widetilde{\mu}_{c}(y))
    \transpose{(\widetilde{\mu}_{c}(y)
      - \sum_{c=1}^{C} \widetilde{\pi}_{c} \widetilde{\mu}_{c}(y))}
    \}.
    \end{split}
    \label{g_def}
\end{equation}
The closed form expression \eqref{g_def} enables the gradient to be
evaluated efficiently by drawing samples from $p(y)$ and computing
direct Monte Carlo integration,
as summarized in Algorithm \ref{alg:a_grad}. We stop the gradient descent iteration whenever the difference between two  conditional mutual information drops below a threshold. 
The conditional mutual information for GMM is given by 
\begin{equation}
\mutualInfo[y_j, j < i]{x}{y_i}
= \entropy[y_j, j < i]{x} - \entropy[y_j, j \leq i]{x}. \label{MI}
\end{equation}
Since the posterior distribution of $x$ conditioned on $\{y_j, j < i\}$ and $x$ conditioned on $\{y_j, j \leq i\}$ are both GMM, an approximation for the entropy of GMM will approximate \eqref{MI}.
Such an approximation is derived in \cite{GMM_chapter}. For GMM described in (\ref{GMM_model})
$
\entropy{x}  \approx  \sum_{i=1}^C \pi_i \log \left((2\pi e)^{n/2}|\Sigma_i|^{1/2}/\pi_i  - (C-1)\right).
$
\yx{This approximation is good when the Gaussian components are not overlapping too much, or more precisely, when $\sum_{c \neq i} \pi_c \mathcal{N}(\mu_c, \Sigma_c)/(\pi_c \mathcal{N}(\mu_c, \Sigma_c)) \ll 1$.}

\section{Proofs}

\begin{proof}[Proof of Lemma \ref{lemma1}]
\yx{%
We will first prove the noiseless case.
The set \(L\) is intended to consist of
at most \(k\) disjoint subsets covering
the part of the signal \(x\) that has not been determined yet.

At each iteration of the loop starting at Line \ref{line:loop},
Algorithm~\ref{alg:bisection} first splits every set in \(L\)
into two of almost equal size, and decides
which of the new sets \(S\) intersects the support of \(S\)
by measuring \(a_{S}^{\intercal} x\).
Then keeps only the \(S\) in \(L\), which intersect the support of
\(x\), and have size greater than \(1\).
On the removed subsets \(S\),
the measurement \(a_{S}^{\intercal} x\) already determines \(x\),
and the estimator \(\widehat{x}\) is updated to coincide with \(x\).

Now we estimate the number of measurements altogether.
As the support of \(x\) has size at most \(k\),
at every iteration \(L\) consist of at most \(k\) sets,
meaning \(2 k\) measurement per iteration.
Finally, due to halving of sets,
as the sizes of the sets in \(L\) are at most
\(2^{\lceil \log n \rceil - i}\) after iteration \(i\),
therefore after at most \(\lceil \log n \rceil\) iteration,
all the sets in \(L\) will have size \(1\),
and the algorithm stops, having determined the whole \(x\).
Thus, at most \(2 k \lceil \log n \rceil\) measurements
are made altogether.

In the noisy case, the main difference is that
every measurement is repeated \(r = \lceil \log n \rceil\) times,
and average is taken
over the block of \(r\) measurements
to reduce the error to at most \(\varepsilon \)
with error probability
at most \(\exp(- r \varepsilon^{2} / 2 \sigma^{2}) / 2\).
Assuming the error is less than \(\varepsilon \)
for every block of measurements,
the algorithm always correctly detects
when a subset \(S\) does not intersect
the support of \(x\),
as then \(y \leq \varepsilon \) in Line~\ref{line:test}.
On such subsets \(x\) is estimated by \(0\),
which is exact.
However, \(y \leq \varepsilon \) might also happen
if \(a_{S} x \leq \varepsilon \) but \(x\) is not \(0\)
on \(S\).
This will not cause \(L\) to consist of more than \(k\)
subsets,
but \(x\) will be estimated by \(0\) on \(S\),
causing errors at most \(\varepsilon \)
on the non-zero coordinates on \(x\) in \(S\). 
Note that Line~\ref{line:estimate}
establishes an error at most \(\varepsilon \)
on subsets \(S\) with \(\card{S} = 1\). 
All in all, the algorithm terminates 
after at most \(2 k \lceil \log n \rceil\)
blocks of measurements,
and in the end
the estimator \(\widehat{x}\) coincides with \(x\)
outside the support of \(x\),
has error at most \(\varepsilon \) in every coordinate
in the support of \(x\).
Therefore \(\norm{\widehat{x} - x}_2 \leq \varepsilon \sqrt{k}\). 
By the union bound,
the probability of making an error greater than
\(\varepsilon \)
in some of the first \(2 k \lceil \log n \rceil\)
blocks is at most \(2 k \lceil \log n \rceil
\exp(- r \varepsilon / (2 \sigma^{2})) / 2\),
which provides the claimed error probability.
}
  \end{proof}

  \begin{proof}[Proof of Lemma \ref{lem:LB-bisection}]
We consider the family \(\mathcal F\) of signals consisting of all \(k\)-sparse
signals on \(n\) bits with uniform distribution. In particular
\(\log |\mathcal F| = \log \binom{n}{k} > \log (\frac{n}{k})^k = k \log
\frac{n}{k}\).
Observe that for every measurement \(a\) we have \(y = \transpose{a}x \in
\face{0,\dots,k}\) and hence the entropy of the measurement result is
less than \(\entropy{y} \leq \log k+1\). 
   We apply Theorem~\ref{thm:estimateRunningTime} to obtain a lower bound on the
   expected  number of measurements \(m\):
\begin{equation}
\begin{split}
\mathbb{E}[m] &\geq \frac{k \log
\frac{n}{k}}{\log k+1}  = \frac{k \log n - k \log k}{\log k+1} \\ & >
\frac{k}{\log k +1} \log n - k > \frac{k}{\log k +1} (-1 + \log n ).
\end{split}
\end{equation}
  \end{proof}

  \begin{proof}[Proof of Lemma \ref{lemma_lower_bound}]
It is easy to observe that after a measurement the size of the domain
is effectively reduced, however the signal \(x\) is still distributed
uniformly at random in the residual set. Thus it suffices to consider
a single measurement. 
Let the measurement $a$ be chosen such that the first half of the entries are \(0\)
and the other half of the entries are \(1\), i.e., we partition \([n]
= A_1 \dot \cup A_2\). The obtained measurement \(Y\) satisfies
$
Y = \left\{
\begin{array}{ll}
0, & \text{w.p. } 1/2; \\
1, & \text{w.p. } 1/2.
\end{array}
\right.
$
Note that \(Y\) is determined by \(X\) given the measurement, i.e.,
\(\entropy[X]{Y} = 0\). We therefore obtain
$\mutualInfo{X}{Y} = \entropy{Y} - \entropy[X]{Y}  = \entropy{Y} = 1$.
On the other hand \(\mutualInfo{X}{Y} \leq \entropy{Y} \leq 1\) as
\(Y\) is binary. Thus the measurement maximizes the mutual
information. 
As the reduced problem after the measurement is identical to the
original one except for the domain size being reduced by a factor of
\(1/2\), by induction, we obtain that the continued bisections maximize the conditional mutual information.
  \end{proof}

\begin{proof}[Proof of Lemma \ref{lem:bisection-k-ub}]
 The considered family \(\mathcal F\) has entropy at least
\(\entropy{\mathcal{F}} = \log |\mathcal F| = \log \binom{n}{k} > \log (\frac{n}{k})^k = k \log
\frac{n}{k}\). On the other hand, the bisection algorithm requires \(k
\lceil\log n\rceil\) queries. Using
Theorem~\ref{thm:estimateRunningTime} and let \(c\) be the upper bound
on information gathered per measurement. We obtain 
\(k
\lceil\log n\rceil \geq (k \log
\frac{n}{k})/c.\) 
Solving for \(c\) we obtain \(c \geq 1- (\log k)/(\log n)\). Thus
the expected amount of information per query is at least \(c\).
\end{proof}

\begin{proof}[Proof of Theorem \ref{thm:Gaussian}]
We consider how the covariance matrix of \(x\) changes
conditioned on the measurements taken.
\yx{As explained in (\ref{eq:covariance-eigenvector}),}
measuring with an eigenvector reduces its eigenvalue
from \(\lambda\) to \(\lambda \sigma^{2} / (\lambda + \sigma^{2})\)
leaving the other eigenvalues unchanged.
Thus, as far as the spectrum of the covariance matrix is concerned,
each measurements applies this reduction to
one of the then-largest eigenvalue. 
Note that Algorithm~\ref{alg:Gaussian-all} might reduce several times
an eigenvalue,
but as mentioned before,
several reductions has the same effect as
one reduction with the combined power.
Thus, to reduce \(\lambda_{i}\) to a value at most
\(\delta \coloneqq \varepsilon^{2}/\chi_{n}^{2}(p)\),
the minimum required number of measurements is
\(\left(1/\delta - 1/\lambda_{i} \right) \sigma^2\)
provided \(\lambda_{i} > \delta\) and \(\sigma > 0\).
Rounding up to integer values and summing up for all directions,
we obtain \eqref{eq:Gaussian-number-noisy}
as a lower bound on total power. 
Furthermore, if \(\sigma^{2} \leq \delta\)
then a single measurement
suffices to ensure \(\lambda_{i} \leq \delta\),
and a measurement is only needed if \(\lambda_{i} > \delta\).
This provides \eqref{eq:Gaussian-number-exact}
even in the noiseless case \(\sigma = 0\). 
All in all,
after the algorithm has finished,
the posterior distribution of the signal \(x\)
is Gaussian \(\Dnormal{\mu'}{\Sigma'}\)
with mean \(\mu'\) and covariance matrix \(\Sigma'\).
The largest eigenvalue \(\norm{\Sigma'}\) of \(\Sigma'\)
is at most \(\delta\),
i.e.,
$
  \varepsilon
  \geq
  \sqrt{\norm{\Sigma'} \cdot \chi_{n}^{2}(p)}.
$
As a consequence,
we show that the mean \(\mu'\) returned by the algorithm
is an estimator of the signal
with the required accuracy \(\varepsilon\).
An easy calculation shows that
the distance between \(x\) and \(\mu'\) is at most \(\varepsilon\)
with probability at least \(p\):
\begin{equation*}
 \begin{split}
  & \probability(x \sim \Dnormal{\mu'}{\Sigma'}){\tnorm{x - \mu'}
    \leq \varepsilon} \\
  &
  \geq
  \probability(x \sim \Dnormal{\mu'}{\Sigma'}){\tnorm{x - \mu'}
    \leq \sqrt{\norm{\Sigma'} \cdot \chi_{n}^{2}(p)}}
  \\
  &
  \geq
  \probability(x \sim \Dnormal{\mu'}{\Sigma'})
  {\transpose{(x - \mu')} \Sigma'^{-1} (x - \mu') \leq \chi_{n}^{2}(p)}
  = p,
 \end{split}
\end{equation*}
where the last equality is a restatement of the well-known
prediction interval for multivariate normal distributions.
\end{proof}

\begin{proof}[Proof of Theorem \ref{thm:Gaussian-fixed-noise}]
The proof is similar to that of Theorem~\ref{thm:Gaussian},
so we point out only the differences.
The power used by Algorithm~\ref{alg:Gaussian-all}
reduced every eigenvalue \(\lambda_{i}\)
to exactly
\(\delta \coloneqq \frac{\varepsilon^{2}}{\chi_{n}^{2}(p)}\),
provided \(\lambda_{i} > \delta\),
otherwise \(\lambda_{i}\) is left intact.
Hence summing up the powers for the eigenvalues,
the total is power is given by
\eqref{eq:Gaussian-number-fix-noise},
and the largest eigenvalue of the posterior covariance matrix
is at most \(\delta\).
The mean is an estimator of the signal with the required accuracy
for the same reasons as in the proof of Theorem~\ref{thm:Gaussian}.
\end{proof}

\begin{proof}[Proof of Theorem \ref{thm:Gaussian-non-iso}]
The proof is similar to Theorem~\ref{thm:Gaussian}.
The only difference is that instead of
the canonical scalar product the one with matrix \(\Sigma_{w}\)
is used.
To make this transparent,
we switch to an orthonormal basis of \(\Sigma_{w}\),
we write \(\Sigma_{w} = \transpose{F} F\),
and use \(F\) as a change of basis:
thus the signal in the new basis is \(F^{-1} x\),
the measurement vectors are \(\transpose{F} a_{i}\),
the covariance matrices are \(\Sigma_{w}^{-1} \Sigma\)
for the signal \(F^{-1} x\),
and the identity matrix for the noise.
In this basis, the algorithm is identical to that of for the white noise added prior to measurement case in
Algorithm~\ref{alg:Gaussian-all},
and hence reduces every eigenvalue of \(\Sigma_{w}^{-1} \Sigma\)
to be at most
\(\norm{\Sigma_{w}}^{-1} \frac{\varepsilon^{2}}{\chi_{n}^{2}(p)}\).
Note that the noise model is
\(y_{i} = \transpose{(\transpose{F} a_{i})} (F^{-1} x) +
\transpose{(\transpose{F} \widehat{a_{i}})} (F^{-1} w)\),
and therefore the power \(\beta_{i}\)
provided by the formula
\(\transpose{F} a_{i} =
\sqrt{\beta_{i}} \transpose{F} \widehat{a_{i}}\),
i.e.,
\(a_{i} = \sqrt{\beta_{i}} \widehat{a_{i}}\).
In other words,
the power \(\beta_{i}\) is still
the length of \(a_{i}\) in the original basis. 
Let \(\Sigma'\) denote the posterior covariance matrix of \(x\)
in the original basis.
Hence \(\norm{\Sigma_{w}^{-1}\Sigma'} \leq
\norm{\Sigma_{w}}^{-1} (\varepsilon^{2}/\chi_{n}^{2}(p))\),
and therefore \(\norm{\Sigma'}\)
is at most \((\varepsilon^{2}/\chi_{n}^{2}(p))\).
This ensures that
the posterior mean of \(x\) returned by the algorithm
to be of the required accuracy, as in Theorem~\ref{thm:Gaussian}.
\end{proof}

\begin{proof}[Sketch of proof for Theorem \ref{thm:GMM}]
We first sequentially run the corresponding algorithm for each
component \(c \in C\). This leads to a number of iteration (or power
consumption) of at most \(\sum_{c \in C} m_c\).  
We perform
measurements that maximize the mutual information between the
signal and the measurement outcome for the mixture.  
With each measurement $a_i$ and the outcome $y_i$, we update the
posterior distribution of the Gaussian component, which we index by \(\pi^i_c\), \(c = 1, \ldots, C\), as follows
\begin{equation}
  \label{eq:4}
\pi_c^{i+1} \coloneqq \pi_c^i \cdot K_i \cdot e^{- \frac{1}{2} \frac{(y_i - 
\transpose{a}_i \mu_{c,i})^2}{\transpose{a}_i
  \Sigma_{c,i}a_i + \sigma^2}},
\end{equation}
where \(\mu_{c, i}\) is
the posterior mean of component \(c\) obtained after the measurement,  and \(K_i\) is a normalization ensuring that
\(\pi^{i+1}\) sum up to 1. The updates in \eqref{eq:4} scale down the
probabilities of those components \(c\) whose mean \(\mu_{c,i}\) leads
to reconstruction with a higher error, which is measured by \(\frac{(y_i - \transpose{a}_i \mu_{c,i})^2}{\transpose{a}_i
  \Sigma_{c,i}a_i + \sigma^2}\). 
Then we apply the hedge version of the multiplicative weight update formula (see e.g.,
\cite[Theorem 2.3]{arora2012multiplicative}) to our setup 
and
obtain that after \(m\) measurements we have
\begin{equation}
\begin{split}
& \frac{1}{m} \left| \sum_{i = 1}^m K_i \cdot \left[ \left(\sum_{\ell = 1}^C \frac{(y_i - \transpose{a}_i \mu_{\ell,i})^2}{\transpose{a}_i
  \Sigma_{\ell,i} a_i + \sigma^2} \cdot \pi_\ell^i \right
) - \frac{(y_i - \transpose{a}_i \mu_{c,i})^2}{\transpose{a}_i
  \Sigma_{c,i}a_i + \sigma^2}\right ] \right| \\
  &  \leq \tilde{\eta} +
\frac{2 \ln \card{C}}{m},
\end{split}
 \nonumber
\end{equation}
for all \(c \in C\). Here $\tilde{\eta} > 0$ is a parameter for multiplicative update algorithm. In particular we can identify the correct component \(c^*\)
whenever \(m = \mathcal{O}(\frac{1}{\tilde{\eta}} \ln \card{C})\).
\end{proof}

\end{document}